\documentclass[review, 3p]{elsarticle}

\usepackage{amsmath}
\usepackage{graphicx}
\usepackage{enumerate}
\usepackage{url}
\usepackage{subfigure}
\usepackage{algorithm}
\usepackage{algpseudocode}
\usepackage{stfloats}
\usepackage{lineno}
\usepackage{hyperref}
\usepackage{bm}
\usepackage{bbm}
\usepackage{enumerate}
\usepackage{xcolor}
\usepackage{float}
\usepackage{tabularx}
\usepackage{tabu}
\usepackage{booktabs}

\newcommand{\bmb}[1]{\bar{\bm{#1}}}
\newcommand{\bmh}[1]{\hat{\bm{#1}}}

\newtheorem{theorem}{Theorem}
\newtheorem{lemma}{Lemma}
\newdefinition{remark}{Remark}
\newdefinition{definition}{Definition}
\newdefinition{assumption}{Assumption}
\newdefinition{motivation}{Motivation}
\newdefinition{philosophy}{Philosophy}
\newdefinition{axiom}{Axiom}
\newproof{proof}{Proof}

\graphicspath{{Figures/}}

\modulolinenumbers[1]

\begin{document}

\begin{frontmatter}

\title{Deep Pattern of Time Series and Its Applications in Estimation, Forecasting, Fault Diagnosis and Target Tracking}

%% Group authors per affiliation:

\author[NUS]{Shixiong Wang \corref{corl}}
\ead{s.wang@u.nus.edu}

\author[NUS]{Chongshou Li}
\ead{iselc@nus.edu.sg}

\author[NUS]{Andrew Lim}
\ead{isealim@nus.edu.sg}

\address[NUS]{Department of Industrial Systems Engineering and Management, National University of Singapore, Singapore 117576}

\cortext[corl]{Corresponding Author.}

\begin{abstract}
The information contained in a time series is more than what the values themselves are. In this paper, the Time-variant Local Autocorrelated Polynomial model with Kalman filter is proposed to model the underlying dynamics of a time series (or signal) and mine the deep pattern of it, except estimating the instantaneous mean function (also known as trend function), including: (1) identifying and predicting the peak and valley values of a time series; (2) reporting and forecasting the current changing pattern (increasing or decreasing pattern of the trend, and how fast it changes). We will show that it is this deep pattern that allows us to make higher-accuracy estimation and forecasting for a time series, to easily detect the anomalies (faults) of a sensor, and to track a highly-maneuvering target.
\end{abstract}


\end{frontmatter}

%\linenumbers

\section{Introduction}  \label{sec:introduction}
\subsection{Subject Matter}
Data is the carrier of the information. Data mining techniques derived from information modeling is of great significance in data science, for example, in signal processing. Since signal is a time series, therefore for generality, we in the following use the terms Time Series and Signal interchangeably. The process of finding the internal mechanism/dynamics of how to generate such a signal by an information system is termed as \textbf{Information Modeling}. This paper is concerned with such modeling for a time series so that we can mine more information contained in the time series.

\subsection{Glossary}
In order not to confuse readers from different communities, we mention the difference of two terms \textit{estimation} and \textit{forecasting}. According to \cite{simon2006optimal} (see chapter 5.1), \textit{estimation} emphasizes the process of estimating the true state values from a set of noised observations, while \textit{forecasting} emphasizes the process of predicting the future values based on the historical values.

\subsection{Literature Review}
To the best knowledge of the authors, the existing literature and methodologies regarding time series analysis belongs mainly to one of the following six categories.

The first category refers to simple methods like average method, (seasonal) naive method, drift method, moving average method, exponential smoothing method, regressions on time as Holt's and Holt-Winters method, and so on \cite{hyndman2008forecasting,hyndman2018forecasting}. Those methods are theoretically easy to understand and practically convenient to adapt into many real problems with acceptable performances. Thus, they are still fashionable at present, especially in engineering.

The second category admits the prestigious Box-Jenkins \cite{box2015time} and autoregressive conditional heteroscedasticity (ARCH) \cite{engle1982autoregressive,brooks2019introductory} methodological families. The Box-Jenkins methodology is also known as ARMA and ARIMA model. The Wold's decomposition theorem \cite{hamilton1995time,papoulis2002probability} is the theoretical basis of Box-Jenkins methodology. As variants of Box-Jenkins methodology, (S)ARIMA and ARMAX models are becoming canonical \cite{hamilton1995time,hyndman2008forecasting,hyndman2018forecasting}. As complements to Box-Jenkins family, the ARCH family take into account the problem of heteroscedasticity. The ARCH family includes many member methods like GARCH \cite{bollerslev1986generalized}, NGARCH \cite{posedel2006analysis}, ZD-GARCH \cite{li2018zd} and Spatial GARCH \cite{otto2018generalised} and so on.

The third category shows interests in Spectral Analysis and Digital Filtering \cite{WANG2019ARIMA,granger1964spectral,koopmans1995spectral,bloomfield2004fourier}. Ref. \cite{granger1964spectral} first conducted the spectral analysis, also known as Fourier Frequency Domain Analysis, to time series analysis. Fourier Series Expansion which is also known as harmonic analysis is particularly popular in this category. This category also admits the Wavelet Transform \cite{daubechies1990wavelet}, Hilbert-Huang Transform \cite{huang2014hilbert} et al., which are extensions of frequency-domain analysis approaches. Notably, Hilbert-Huang transform is prestigious for its data-adaptivity presented in Empirical Mode Decomposition method, meaning the transformation basis is not fixed in advance, which is rather different from the Fourier transform and wavelet transform.

The fourth category includes the reputed Adaptive Filters, like Kalman Filter family \cite{simon2006optimal,hamilton1995time,hyndman2008forecasting}, Information Filter \cite{hyndman2008forecasting,chandra2013square} and so on.  The Kalman Filter family is mainly well studied and utilised in Signal Processing community especially like the target tracking field \cite{li2005survey}. The traditional Kalman filter aims to give the unbiased minimum variance estimates to the system states for a linear and white (uncorrelated) noise system. However, for the general time series analysis community like finance time series field, traffic volume prediction field et al., the dilemma of using the Kalman filter is that it is hard and/or even impossible to obtain the claimed State Equation, an equation analytically describing the dynamics of the focused time series. Worthy of mentioning is that Refs. \cite{hyndman2008forecasting} and \cite{durbin2012time} used the state space method to bridge the gap between Kalman filter and some time series analysis methods like local level method, exponential smoothing method, Holt's method (known as linear trend model) and the like. Note that the Kalman Filter(s) conceptually denote(s) a family of adaptive filters based on the traditional Kalman filter, like Extended Kalman Filter \cite{wang2015nonlinear}, Unscented Kalman Filter \cite{sarkka2010gaussian}, Cubature Kalman Filter \cite{wang2015nonlinear,sarkka2010gaussian}, Tobit Kalman Filter \cite{geng2019state}, and many other Improved Kalman Filters \cite{liang2004finite,sarkka2010gaussian}.

The fifth category focuses on the newly boosting Machine Learning methods, in computer science. In this category, two sub-categories should be paid attention to: (a) Kernel based method. Note that all the aforementioned four categories model the noise in a time series as a stochastic process. However, there also exist some other theories that regard a time series as a chaotic process \cite{liu2011kernel}. For the chaotic process model, the kernel based methods stand out. Those methods are like kernel adaptive filters \cite{ma2017robust}, kernel affine projection (KAP) algorithm \cite{richard2009online,liu2011kernel}, kernel recursive least squares algorithm \cite{han2018multivariate,liu2011kernel}, and kernel least mean kurtosis based method \cite{hua2013kernel}. As reported, the kernel based methods have amazing performances in predicting some complex time series. (b) Deep Learning method. Another powerful method from machine learning community to forecast a time series is Deep Learning \cite{zhu2018power,liao2018deep,siami2018forecasting}. As they stated \cite{zhu2018power,liao2018deep,siami2018forecasting,lu2017agent}, there exist many advantages of deep learning over other methods like: (i) easy to extract features of a time series and further make satisfying prediction; (ii) recurrent Neural Networks (RNN) and Long Short Time Memory Network (LSTM) have inborn power to identify and express the underlying dynamics/patterns of a time series which evolves along the time; (iii) allowing the multiple inputs and multiple outputs scenarios.

The sixth category admits the combined methods like Vector auto-regression \cite{hyndman2018forecasting}, TBATS \cite{de2011forecasting,hyndman2018forecasting}, ARMA-SIN \cite{WANG2019ARIMA} and so on. Specifically, Vector auto-regression is a special case of ARMAX model with multiple exogenous variables; TBATS is a combination of Fourier Series Expansion, Exponential Smoothing, State Space Model, and Box-Cox transformation. For more examples on this point, please see \cite{hyndman2018forecasting,wang2012stock,doya2002multiple,murray1997multiple}.
%Another worthy-mentioning combined method is presented by \cite{wang2012stock}, in which the exponential smoothing method, ARIMA model and back propagating neural network are combined together to forecast the stock market price.

\subsection{Research Gap and Main Contributions}
As we can see above, scholars in today's research only pay attention to understanding the historical values of a time series and predicting the future values based on the historical values, giving no emphasis on the changing pattern (i.e., increasing, decreasing, and how large/fast the change is) and extrema values (peaks and/or valleys) from viewpoint of modeling. The changing pattern, as well as the extrema values and where they occurs, other than the value themselves are also very important to investigate, for example, the Stock Understanding and Forecasting problem. \textit{\textbf{In other words, the information contained in a data series is more than what the values themselves of it are}}. Thus, the \textbf{Deep Patterns} of a data series should be further studied. To this end, in this paper, we will firstly model a time series (signal) as a non-stationary stochastic process presenting the properties of variant mean. Then the Time-Variant Local Autocorrelated Polynomial Model with Kalman Filter will be proposed to dynamically estimate the deep patterns of the focused time series.
\begin{remark}
We use the word \textit{deep} here for the reason that such deep information (deep pattern) of a time series is not obvious and straightforward to numerically obtain from the raw values of the time series. The mathematical essence of deep pattern will be revealed later in Subsection \ref{subsec:deep-pattern-essence}.
\end{remark}

The advantages (main contributions) of our method embody:
\begin{enumerate}[(a)]
\item \textbf{General Modeling for Kalman Filter}. The choose and use of Kalman filter requires that we know the system dynamics in advance. However, for a general signal (time series), this is impossible to satisfy,  because we have no knowledge of the system generating the focused time series. Fortunately, our Time-variant Local Autocorrelated Polynomial (TVLAP) model could help detour this dilemma;
\item \textbf{Trend Estimation}. Extracting the instantaneous mean function (trend function) of a time series, in online manner;
\item \textbf{Extrema Detecting}. Automatically identifying and predicting the peak and valley values of a time series;
\item \textbf{Trend Tracking}. Automatically reporting and forecasting the current changing pattern (increasing or decreasing pattern of the trend, and how fast it changes);
\item \textbf{Real Time}. Our method is workable for the sequential data, not just block data, as exponential smoothing and moving average method do;
\item \textbf{Theoretical sufficiency}. The reliability guarantee of our method is derived.
\end{enumerate}

In the end, we will show some potential applications of the proposed method:
\begin{enumerate}[(a)]
\item \textbf{Estimation and Forecasting}. It is the use of deep pattern that makes more satisfactory estimation and forecasting for a time series;
\item \textbf{Fault Dignosis/Anomaly Detection}. Sensors sometimes unavoidably suffer from anomalies (faults). We aim to use deep patten of the measurements to easily identify those anomalies, and decide whether the sensor is currently reliable or not;
\item \textbf{Highly-maneuvering Target Tracking}. We will show that the deep pattern of the measurements from sensors gives us the possibility to track a highly-maneuvering target.
\end{enumerate}

\section{Preliminary on Stochastic Process}\label{sec:preliminary}
%Recall that in Introduction \ref{sec:introduction} we have already point out the philosophy of ARMA and (S)ARIMA. They in fact aim to model the dynamics of a stochastic process so that we can use the past information from collected time series to satisfyingly predict the future. Mentioning this, we cannot ignore the reputed Wold's Decomposition Theorem in stochastic process analysis.
%
%Firstly we give the strict mathematical definition of the Wide-sense stationary (WSS) stochastic process.
We use $x(t)$ to denote a continuous time stochastic process and $x(n)$ a discrete time one, meaning $t = T_s n$, if the sampling time period is $T_s$.

\begin{theorem}[Wold's Decomposition Theorem \cite{papoulis2002probability}]\label{thm:wold-decom}
Any wide-sense stationary (WSS) stochastic process $x(n)$ could be decomposed into two sub-processes: (a) Regular process; and (b) Predictable process. Namely
%\begin{equation}\label{eq:wold-decom}
  $x(n) = x_r(n) + x_p(n)$,
%\end{equation}
where $x_r(n)$ is a regular process and $x_p(n)$ is a predictable process. Furthermore, the two processes are orthogonal (meaning uncorrelated): $E\{x_r(n+\tau)x_p(n)\} = 0$.
\end{theorem}

The detailed concepts of Regular process and Predictable process could be found in \cite{papoulis2002probability}. Intuitively, a WSS stochastic process is mathematically as $x = ARMA(p,q)$ (an autoregressive moving average process with autoregressive order of $p$ and moving average order of $q$). Thus, this theorem reveals the philosophy of Box-Jenkins methodology \cite{box2015time}.

\section{Problem Formulation and Motivations}\label{sec:formulation}
%In this section, we aim to mathematically formulate the problem concerning the issues we raised in this paper.
\subsection{Notations}
%Before proceeding to the details, we introduce mathematical notations here.
\begin{enumerate}
  \item Let $\bm v = a:l:b$ define a vector $\bm v$ being with the lower bound $a$, upper bound $b$ and step length $l$. For example, $\bm v = 0:0.1:0.5$ means $a=0$, $b=0.5$, and $l=0.1$. Thus $\bm v = [0, 0.1, 0.2, 0.3, 0.4, 0.5]^T$;
  \item Let the function $length({\bm x})$ return the length of the vector $\bm x$. For example, if $\bm x = [1,2,3]^T$, we have $length({\bm x}) = 3$;
  \item Let $\bm t$ denote the continuous time variable, and $\bm n$ its corresponding discrete time variable. For example, if $\bm t = 0:0.5:100$ (the time span is $100s$, and the sampling time is $T_s = 0.5s$), we will have $\bm n = t/{T_s} = 0:1:(length(\bm t)-1) = 0:1:200$; Let $N = length(\bm n)$. For notation simplicity, we use $T$ and $T_s$, interchangeably;
  \item Let ${\bm x}(n)$ or ${\bm x}_n$ denote the interested time series, shorted as $\bm x$; Let $\bm y$ denote the transformed time series from $\bm x$;
  %\item Let ${\bm x}_0$ denote the exact dynamics of $\bm x$, ${\bmh x}$ the estimated dynamics of $\bm x$. It means $\bm x$ is generated from its ground truth ${\bm x}_0$. The closer between the transformed (estimated) series ${\bmh x}$ and ${\bm x}_0$, the better the transform and the more proper the ARMA model could be used to refactor the dynamics of $\bmh x$;
  \item Let the function $mean(\bm x)$ return the mean of a random variable $\bm x$, and $var(\bm x)$ the variance of it. If $\bm x$ is a stochastic process, then $mean(\bm x)$ denotes the mean function and $var(\bm x)$ the variance function;
  \item Let ${\bm G}'$ denote the transpose of the matrix $\bm G$;
  \item Let the operator $ARMA(p,q|\bm {\varphi}, \bm {\theta})$ denote a ARMA process with autoregressive order of $p$ and moving average order of $q$. Besides, the coefficient vectors $\bm {\varphi}$ and $\bm {\theta}$ are for autoregressive part and moving average part, respectively. $ARMA(p,q|\bm {\varphi}, \bm {\theta})$ is shorted as $ARMA(p,q)$.
\end{enumerate}

\subsection{General Model of a Non-stationary Stochastic Process}\label{subsec:general-model-non-stationary}
In this paper, we consider a general model describing a non-stationary stochastic process with the following form
\begin{equation}\label{eq:general-model-lite}
  x(n) = f(n) + x_s(n),
\end{equation}
where $x_s(n) := ARMA(p,q)$ is a WSS stochastic process; $f(n)$ is a deterministic function denoting the mean function of the time series $x(n)$. Note that the expectation of the term $x_s(n)$ is zero. Therefore, we have $mean(x) = f$, and $var(x) = var(x_s)$.

\subsection{Deep Pattern of a Time Series}\label{subsec:deep-pattern-essence}
As a snapshot, we mention that the mathematical essences of Deep Patterns we mention in this paper are high-order derivatives of $f(n)$. Therefore, the ultimate purpose of this paper is to estimate $f(n)$ and its high-order derivatives from noised $x(n)$, \textbf{\textit{in online manner}}. Note that the changing pattern of a time series is reflected by the first-order derivative; the extrema points are reflected by the both first-order and second-order derivative. The motivation is to regress $f(n)$ with an order-sufficient polynomial, in online manner.

\subsection{Problem Statement}
In summary, $x(n)$ is our signal model; the key information contained in $x(n)$ is modelled as $f(n)$; the part $x_s(n)$ is the model of noise contained in the signal $x(n)$. We in this paper aim to recover the $f(n)$ and its high-order derivatives from noised $x(n)$.

\subsection{General Motivations}
We in Motivation \ref{motiv:general-idea-real-time} disclose the general idea for: (1) the trend estimating and tacking problem; and (2) the minima or maxima detecting ({extrema detecting}) problem.

\begin{motivation}\label{motiv:general-idea-real-time}
Any continuous function could be approximated by a polynomial function with sufficient orders. Plus, the polynomial functions are high-order differentiable, meaning we can assert a point to be a: (a) minimum if the first-order derivative is zero-valued and second-order derivative is positive, (b) maximum if the first-order derivative is zero-valued and second-order derivative is negative, at this point. Besides, the real-time changing pattern of trend could also be handled by the first-order derivative, increasing if positive, or decreasing if negative. Thus, we desire to use an order-sufficient polynomial to regress the time series of interest in online manner.
\end{motivation}

However, the dilemmas are that existing polynomial regression methods: (a) like traditional global polynomial regression only works for block data, fails to work for real-time scenario; (b) like Holt's method is order-deficiency, only holding the first-order (linear trend) polynomial which introduces potential time-delay problem when the changing rate of the focused time series is sharp. The issue of time-delay also exists in non-polynomial methods like exponential smoothing and moving average, reported by \cite{WANG2019ARIMA}. Thus we aim to in this paper introduce the Time-Variant Local Autocorrelated Polynomial Model with Kalman Filter (TVLAP-KF) to handle this issue.

As for the problem of the non-stationarity of the focused time series, we have a general idea to settle this in Motivation \ref{motiv:general-idea-regularization}.
\begin{motivation}\label{motiv:general-idea-regularization}
Facing the non-stationarity problem (\ref{eq:general-model-lite}), if we can estimate out the function $f(n)$, we could then have an estimated $\hat{x}_s (n) = [x(n) - \hat{f}(n)]$ which is considered to be a ARMA process. Here, $\hat{f}(n)$  is the estimate to $f(n)$. Then we can apply Box-Jenkins methodology to handle the stationary residual $\hat x_s(n)$, and the non-stationarity problem is solved.
\end{motivation}

%The existing methods like exponential smoothing, moving average, Holt's method, TVLAP-KF method in this paper and so on could identify and eliminate the low-frequency trend component \cite{WANG2019ARIMA} of a time series, to some degree, in real time. Thus it seems not challenging to estimate the $f(n)$.

%\begin{remark}\label{rem:validity-ETS-MA}
%Note that the theoretical validity and improvements of exponential smoothing method and moving average method in detrending have been proved by \cite{WANG2019ARIMA}.
%\end{remark}
%
%\begin{remark}\label{rem:if-concern-periodicity}
%If one also has interests in periodicity, meaning he concerns the model (\ref{eq:general-model-trans}) rather than (\ref{eq:general-model-lite}), he can use the Spectral Analysis and Digital Filtering methodology (the SIN part of ARMA-SIN methodology) to identify and extract the periodic component of the interested time series. After this, he can continue to implement the ideas in Motivation \ref{motiv:general-idea-real-time} and Motivation \ref{motiv:general-idea-regularization}.
%\end{remark}

\section{Trend Tracking and Extrema Detecting}\label{sec:polynomial-model}
\subsection{Time-Variant Local Autocorrelated Polynomial Model}
In this section, we will introduce the Time-Variant Local Autocorrelated Polynomial (TVLAP) Model with Kalman Filter to handle the online polynomial regression problem (mean estimation) and extrema detection problem. As a demonstration, we in this section only takes the special case of (\ref{eq:general-model-lite}) as $x(n) = f(n) + white(n)$, where $white(n)$ denotes a white-noise (uncorrelated) series. Plus, readers are invited to refer to \cite{simon2006optimal} (see chapter 5.1) for profoundly understanding about Kalman filter, including the estimation method and the prediction method. In consideration of paper length and necessity, we will not introduce more about the theory of Kalman filter.

As we state in Introduction, the Kalman filter is powerful only when the required state equation (also known as system equation, system dynamics equation, or transfer function in state space et al. in control theory and signal processing community) is known. This limits the wide utilization of Kalman filter in time series analysis and signal processing, because generally we cannot know the explicit dynamics (state equation) of an information system which generates the focused time series, unlike many problems in control theory and signal processing. However, the Kalman filter is extremely attractive for us due to that: (a) it is an online algorithm with low computational complexity; (b) it is an optimal linear estimation method in linear-system and white-noise sense. Since we are concerned with the trend estimating and trend tracking problem, and the mean of a time series is generally low-frequency, why not to use an order-sufficient polynomial to refactor the dynamics of the mean of the time series, namely, the $f(n)$ part in (\ref{eq:general-model-lite})? That is to say, why do not we model the information dynamics of the time series?

The theoretical validity and sufficiency of polynomial regression is from the prestigious Weierstrass approximation theorem \cite{kreyszig1978introductory}. However, the dilemma is the real-time and extrema detecting issues, meaning we expect the algorithm to be able to not only work online but also simultaneously return the current first-order and second-order (or even higher-orders) derivatives of such a well-approximated polynomial. Well-approximation here means the regressed polynomial and the mean function of the raw time series are close enough.

Fortunately, the dilemma is possible to detour when we ask for help from the Taylor's expansion, asserting that a real-valued function $f(t)$ that is infinitely differentiable at a real number $t_0$ could be the power series with the form of %(\ref{eq:taylor-expansion}),
%\begin{figure*}[htp]
%\normalsize
\begin{equation}\label{eq:taylor-expansion}
  f(t) = f(t_0) + \frac{f^{(1)}(t_0)}{1!} (t - t_0) + ... + \frac{f^{(k)}(t_0)}{k!} (t - t_0)^k + ...,
\end{equation}
%    \begin{equation}\label{eq:taylor-expansion-discrete-time}
%      f(t) = f(n) + \frac{f^{(1)}(n)}{1!} (t - n) + \frac{f^{(2)}(n)}{2!} (t - n)^2 + ... + \frac{f^{(k)}(n)}{k!} (t - n)^k + ....
%    \end{equation}
%\hrulefill
%\end{figure*}
where $f^{(k)}(t_0)$ denotes the $k^{th}$-order derivative of $f(t)$ at $t_0$.

Note that (\ref{eq:taylor-expansion}) is also a polynomial with a special mathematical form rather than its general form below
\begin{equation}\label{eq:general-polynomial}
  f(t) = f_0 + f_1 t +  f_2 t^2 + ... + f_k t^k + ...,
\end{equation}
where $f_k$ are constant coefficients.

However, a function could be expanded as Taylor's series if and only if it is infinitely smooth, meaning infinitely differentiable. Thus we cannot directly apply the Taylor's series expansion over a general time series whose trend function $f(t)$ may be discontinuous in (high-order) derivatives. To overcome this, we introduce an intermediate (temporary) function $p(t)$ as the Weierstrass approximation of $f(t)$. It means $p(t)$ is a polynomial with proper orders. Thus, we have $\forall \varepsilon >0 $, $\exists \bar K  > 0$, such that
\begin{equation}\label{eq:Weierstrass-approximation}
  \displaystyle \sup_{t}|f(t)-p_{\bar K}(t)| < \varepsilon,
\end{equation}
over a compact space, where $p_{\bar K}(t)$ denotes the polynomial with order of $\bar K$. For simplicity, we ignore $\bar K$ in notation. We have
%\begin{equation}\label{eq:general-polynomial-with_p}
  $p(t) = p_0 + p_1 t +  p_2 t^2 + ... + p_k t^k + ...$.
%\end{equation}

Thus, when we have a time series $x(n)$, we could alternatively choose the polynomial in Taylor's form to regress $p(t)$ instead of $f(t)$ because only $p(t)$ is guaranteed to be infinitely differentiable. This will not lead to disaster, according to (\ref{eq:Weierstrass-approximation}). Suppose we have interests in the properties at the discrete time index $n$, Eq. (\ref{eq:taylor-expansion}) could then be rewritten as (\ref{eq:taylor-expansion-discrete-time}).
\begin{equation}\label{eq:taylor-expansion-discrete-time}
  p(t) = p(n) + \frac{p^{(1)}(n)}{1!} (t - n) + ... + \frac{p^{(k)}(n)}{k!} (t - n)^k + ....
\end{equation}
%\begin{figure*}[htp]
%\normalsize
%    \begin{equation}\label{eq:taylor-expansion-discrete-time}
%      f(t) = f(n) + \frac{f^{(1)}(n)}{1!} (t - n) + \frac{f^{(2)}(n)}{2!} (t - n)^2 + ... + \frac{f^{(k)}(n)}{k!} (t - n)^k + ....
%    \end{equation}
%\hrulefill
%\end{figure*}

Thus, the traditional polynomial regression (\ref{eq:general-polynomial}) could be regarded as the special case of (\ref{eq:taylor-expansion}) when we investigate the problem from the starting point of the time, namely, $t_0 = 0$. In other words, the polynomial (\ref{eq:taylor-expansion-discrete-time}) is a local polynomial, while (\ref{eq:general-polynomial}) is a global polynomial. For intuitive understanding, see Fig. \ref{fig:local-global-polynomial}.

\begin{figure}[htbp]
    \centering
    \subfigure[Global polynomial]{
        \begin{minipage}[htbp]{0.46\linewidth}
            \centering
            \includegraphics[height=4.5cm]{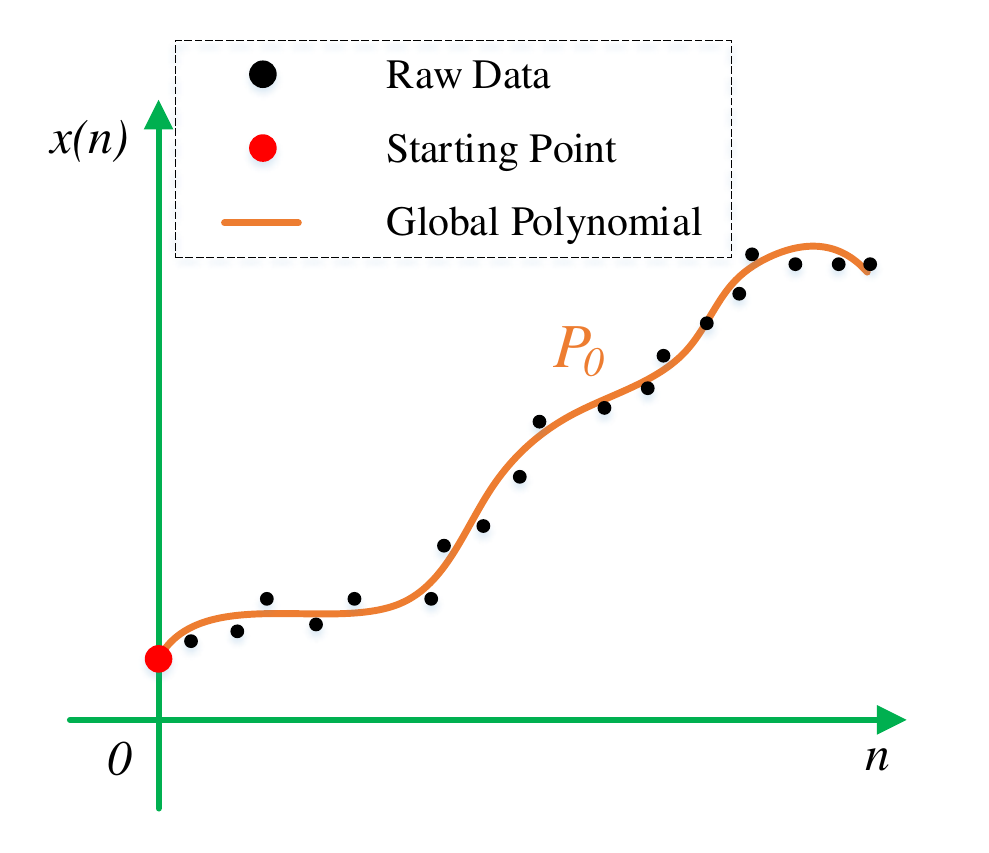}
        \end{minipage}
    }
    \subfigure[Local polynomial]{
        \begin{minipage}[htbp]{0.46\linewidth}
            \centering
            \includegraphics[height=4.5cm]{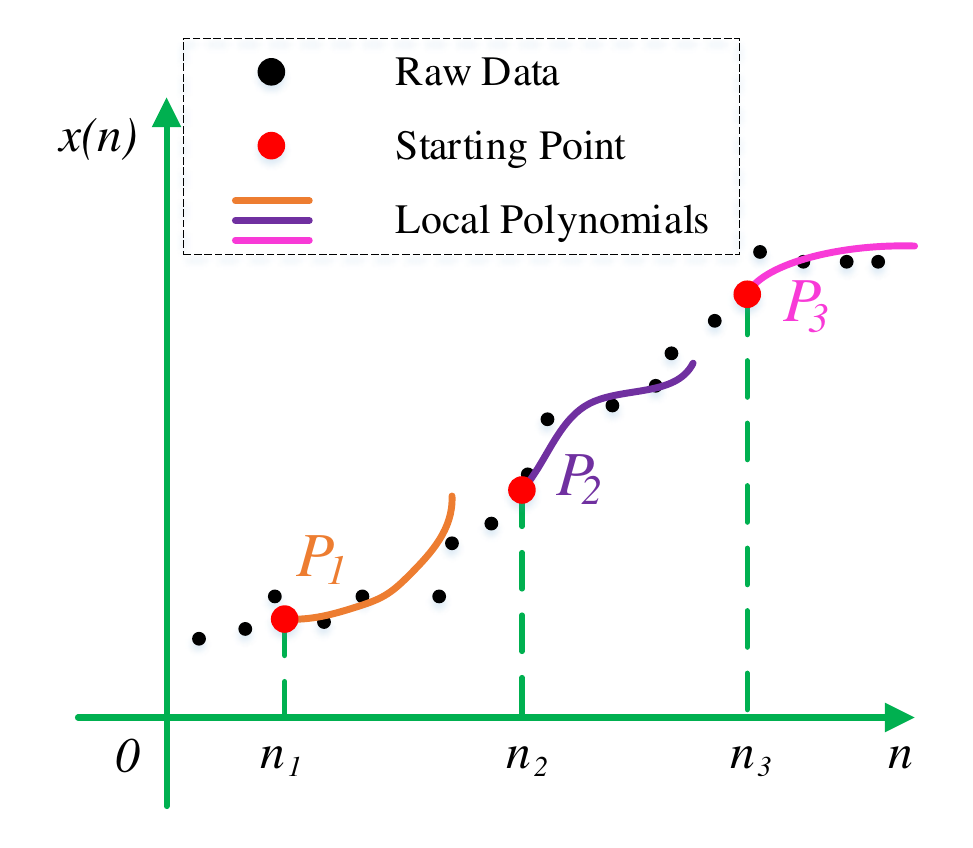}
        \end{minipage}
    }
    \caption{Global and local polynomial}
    \label{fig:local-global-polynomial}
\end{figure}

If we only pay attention to the case of $t = n+1$ and truncate the polynomial on the order of $K$, we have (\ref{eq:taylor-expansion-discrete-time}) as
\begin{equation}\label{eq:state-equation}
    %\begin{array}{cllll}
        p(n+1) %&= f(n) + \displaystyle \frac{f^{(1)}(n)}{1!} T + ... + \displaystyle \frac{f^{(K)}(n)}{K!} T^K \\
               = \displaystyle \sum_{k = 0}^{K} \frac{p^{(k)}(n)}{k!} T^k = \displaystyle \sum_{k = 0}^{K} \frac{T^k}{k!} p^{(k)}(n),
    %\end{array}
\end{equation}
where $T$ denotes the time slot between the discrete time index $n+1$ and $n$.

Interestingly, Eq. (\ref{eq:state-equation}) holds the following powerful characteristics:
\begin{enumerate}[(a)]
  \item It is actually the required State Equation in Kalman filter for a general time series. Note that the nature of the state equation is the recursive relationship of a time-related function from the former discrete time index $n$ to the latter $n+1$;
  \item It conveys the high-order derivatives up to the order of $K^{th}$ of the function $p(t)$, which is attractive in extrema detection and other analysis.
\end{enumerate}

Now it is possible for us to apply the Kalman filter as long as we could have the state space representation of (\ref{eq:state-equation}). In consideration of the fact that the terms $p^{(k)}(n),~k = 0,1,2,3,...,K$ actually change through the time and convey the explicit physical meanings of $p(n)$ (the complete changing patterns), we could choose them as our state variables. Note that only a variable instead of a constant could be considered as a state variable. Thus we define our state vector as
\begin{equation}\label{eq:state-vector}
    \bm X(n) :=  \left[
                    \begin{array}{l}
                        X_0(n) \\
                        X_1(n) \\
                        X_2(n) \\
                        \cdots \\
                        X_K(n)
                    \end{array}
                \right]
             :=  \left[
                    \begin{array}{l}
                        p^{(0)}(n) \\
                        p^{(1)}(n) \\
                        p^{(2)}(n) \\
                        \cdots \\
                        p^{(K)}(n)
                    \end{array}
                \right],
\end{equation}
meaning the first entry is the real-time value of $p(n)$ and the rest entries are the real-time values of the high-order derivatives of $p(n)$.

Consequently, we have the state space representation of (\ref{eq:state-equation}) as
\begin{equation}\label{eq:system-equation}
    \bm X(n+1) =  \left[
                    \begin{array}{*{45}{c}}
                        {1}&{T}&{\frac{T^2}{2}}&{\cdots}&{\frac{T^K}{K!}}\\
                        {0}&{1}&{T}&{\cdots}&{\frac{T^{K-1}}{(K-1)!}}\\
                        {0}&{0}&{1}&{\cdots}&{\frac{T^{K-2}}{(K-2)!}}\\
                        {\vdots}&{\vdots}&{\vdots}&{\ddots}&{\vdots}\\
                        {0}&{0}&{0}&{\cdots}&{1}
                    \end{array}
                \right]
                \bm X(n).
\end{equation}

Eq. (\ref{eq:system-equation}) implies that when we model the dynamics of $f(n)$, we actually admit the $K^{th}$-order derivative to remain constant over time. This is therefore the cost of truncation, meaning we just use the Level Model ($0$-order holder, $0$-order local polynomial) to smooth the $K^{th}$-order derivative. To be more specific, Level model models the slow-changing pattern of the $K^{th}$-order derivative rather than fast-changing pattern. For more on this point, see the philosophy, implementation and performances of Level model introduced in \cite{durbin2012time}.

In time series analysis scenario, only the sequential data $x(n)$ is obtainable, observable. Therefore, we in our state space adaptation should define the measure vector as
\begin{equation}\label{eq:measure-vector}
\begin{array}{rll}
    \bm Y(n) &= x(n) = f(n) + x_s(n) = p(n) + x_s(n) \\
            &:= f(n) + white(n).
\end{array}
\end{equation}

By doing so, we have the Measurement Equation (also known as Observation Equation, or Output Equation) as
\begin{equation}\label{eq:measure-equation}
    \bm Y(n) :=  \left[
                    \begin{array}{cccccc}
                        {1}&{0}&{0}&{\cdots}&{0}
                    \end{array}
                \right] \bm X(n) + \bm V(n),
\end{equation}
where $\bm V(n)$ is used to model the $white(n)$ part (in general, the $x_s(n)$).

Besides, let's define
\begin{equation}\label{eq:system-matrix}
    \bm \Phi :=  \left[
                    \begin{array}{*{45}{c}}
                        {1}&{T}&{\frac{T^2}{2}}&{\cdots}&{\frac{T^K}{K!}}\\
                        {0}&{1}&{T}&{\cdots}&{\frac{T^{K-1}}{(K-1)!}}\\
                        {0}&{0}&{1}&{\cdots}&{\frac{T^{K-2}}{(K-2)!}}\\
                        {\vdots}&{\vdots}&{\vdots}&{\ddots}&{\vdots}\\
                        {0}&{0}&{0}&{\cdots}&{1}
                    \end{array}
                \right],
\end{equation}
as our System Matrix in Kalman filter, and
\begin{equation}\label{eq:measure-matrix}
    \bm H :=  \left[
                    \begin{array}{cccccc}
                        {1}&{0}&{0}&{\cdots}&{0}
                    \end{array}
                \right],
\end{equation}
as our Measurement Matrix. Note that the $\bm \Phi$ and $\bm H$ are constant if given the order $K$. Suppose the state noise vector is $\bm W (n)$ with covariance $\bm Q(n)$ and its noise-driven matrix is $\bm G$; the measurement noise state vector is $\bm V(n)$ with covariance $\bm R(n)$, we then have a state-space model as a stochastic process to a time series $x(n) = p(n) + white(n) = f(n) + white(n)$ as
\begin{equation}\label{eq:linear-system-x}
    \left\{ \begin{array}{rll}
        \bm X(n+1) &= \bm \Phi \bm X(n) + \bm G \bm W(n) \\
        \bm Y(n) &= \bm H \bm X(n) + \bm V(n),
    \end{array} \right.
\end{equation}
where $\bm W(n)$ denotes the modeling error.
%Specifically, it models the difference between (\ref{eq:system-equation}) and (\ref{eq:general-polynomial-with_p}), that is, between (\ref{eq:state-equation}) and (\ref{eq:general-polynomial-with_p}).

Eq. (\ref{eq:linear-system-x}) is the linear system model of $x(n)$ in state space. Note that the mathematical form of $\bm G$ is not unique, meaning we can define it as any proper one. Some simple examples are: (a) $\bm G = [\frac{T^K}{K!},...,T,1]'$ so that $\bm W(n)$ should be a $1$-dimensional vector denoting the disturbance exerted to $X_K(n)$; (b) $\bm G = diag\{\frac{T^K}{K!},...,T,1\}$ so that $\bm W(n)$ should be a $(K+1)$-dimensional vector denoting the disturbance exerted to $\bm X(n)$; (c) $\bm G$ as an identity matrix so that $\bm W(n)$ should be a $(K+1)$-dimensional vector denoting the disturbance exerted to $\bm X(n)$. The difference between (b) and (c) is reflected in their corresponding $\bm Q(n)$.

For our model (\ref{eq:general-model-lite}) studied in this paper, its noise part $x_s(n)$ is stationary, meaning $\bm R(n)$ is constant over time. Let $\bm R := \bm R(n)$.
\begin{remark}\label{rem:enveloped-noise}
  If we allow the heteroscedasticity in noise, we could re-model (\ref{eq:general-model-lite}) as
  %\begin{equation}\label{eq:enveloped-model}
    $x(n) = f(n) + g(n)x_s(n)$,
  %\end{equation}
  where $g(n)$ is another deterministic function. However, this problem is more complicated to handle. Therefore, we leave it as an open problem.
\end{remark}

\subsection{Estimate the \textbf{R} and \textbf{Q}(n)}\label{subsec:estimate-R}
Actually, $\bm R$ is easy to estimate from the historical observations (measures) of $x(\bar n)$. Here $\bar n$ is used to differentiate from $n$, meaning $x(\bar n)$ could be any segment of $x(n)$ in the past, just as ground truth to estimate $\bm R$. Suppose we use the traditional global polynomial $\bar p(t)$ to fit $x(\bar n)$, we should have the fitting residual $\delta$ as
  %\begin{equation}\label{eq:global-poly-fit-residual}
    $\delta(\bar n) := x(\bar n) - \bar p(\bar n)$.
  %\end{equation}
According to our model assumption, $\delta(\bar n)$ should be a wide-sense stationary (WSS) stochastic process, meaning the selected order of $\bar p(t)$ is proper if and only if $\delta(\bar n)$ is wide-sense stationary.

Thus we have
%\begin{equation}\label{eq:estimate-measure-covariance}
$R := var(\delta)$.
%\end{equation}
Note that in (\ref{eq:linear-system-x}), $\bm V$ is 1-dimensional, meaning $\bm R$ is also a scalar rather than a vector, denoted as $R$.

As for the real-time, adaptive estimate to $\bm Q(n)$ when given (\ref{eq:linear-system-x}) and $\bm R$, it is a bit more complicated. Readers are invited to refer to \cite{moghe2019adaptive,mehra1972approaches,mohamed1999adaptive} and other similar literature concerning the process covariance estimation problem. Actually, according to some optimal filtering techniques treating $\bm Q(n)$ as unknown disturbances or unknown inputs \cite{liang2004finite,myers1976adaptive,liang2008adaptive,xia1994adaptive}, it is not always necessary for us to estimate $\bm Q(n)$, because we are only concerned with the optimal estimate of $\bm X(n)$.
\begin{remark}
 In practice, at times there is no need to pursue the exactly true value of $\bm Q(n)$. Engineers could try different $\bm Q(n)$ to obtain different estimation performances. Note that the value of $\bm Q(n)$ actually adjust our trust level towards the system model we use \cite{simon2006optimal}.
\end{remark}

\subsection{TVLAP Model with Kalman Filter}
Now, it is sufficient to use the Kalman filter to handle the linear system (\ref{eq:linear-system-x}), during which we could also estimate the real-time value $\hat{X}_0$ of $p(n)$, and real-time values of $k^{th}$-order derivative $\hat{X}_k$ of $p(n)$, where $p(n)$ is the mean function of the focused time series $x(n)$. The estimates to derivatives admit the feasibility of extrema detecting and the changing-pattern prediction (to predict the increasing pattern or decreasing pattern).

We in this paper term the presented method as Time-Variant Local Autocorrelated Polynomial Model (TVLAP) with Kalman Filter, shorted as TVLAP-KF. \textit{Time-Variant} means the coefficients of the used polynomial model (\ref{eq:state-equation}), namely $p^{(k)}(n)/{k!}$ and $\hat{X}_k(n)$, change over time. The meaning of the word \textit{Local} has been explained earlier in Fig. \ref{fig:local-global-polynomial}. \textit{Autocorrelated} means the coefficients of the used polynomial are not independent, are instead highly related, because we have
\begin{equation}\label{eq:autocorrelated-coefficients}
    \displaystyle \frac{p^{(k)}(n)}{k!} = \displaystyle \frac{d\left[\displaystyle \frac{p^{(k+1)}(n)}{(k+1)!}\right]}{dt}.
\end{equation}

Finally we present the entire algorithm of online trend tracking and extrema detecting in Algorithm \ref{algo:Mean-Extrema-Estimating}.
\begin{algorithm}[htbp] \label{algo:Mean-Extrema-Estimating}
	\caption{Online Trend Tracking and Extrema Detecting for a Variant Mean and White Time Series}
	\begin{flushleft}
        \textbf{Definition}: $\bm P$ as state estimate covariance in Kalman filter; $\bm I$ as identity matrix with proper dimension; $\infty$ as a big number; $\epsilon$ as a small number; $\text{abs}(x)$ as the absolute function which return the absolute value of a real number; $\emptyset$ as an empty set

        \textbf{Reservation}: Set $\mathbbm{E}_m$ to record minima, and Set $\mathbbm{E}^m$ to record maxima

    	\textbf{Initialize}: $\infty \leftarrow 10^5$, $\epsilon \leftarrow 10^{-6}$, $\bm X \leftarrow \bm 0$, $\bm P \leftarrow \infty \times \bm I$, $\bm Q$, $\bm R$, $\mathbbm{E}_m \leftarrow  \emptyset$, $\mathbbm{E}^m \leftarrow  \emptyset$
    \end{flushleft}
	\begin{algorithmic}[1]
		\Require $x(n)$ , $n=0,1,2,3,...$
		\While {true}
    		\State $n \leftarrow n + 1$
            \State // Mean Estimating and Trend Tracking
    		\State $\bmh X (n)$ = Kalman\_Filter[$x(n)$]~~~~// See \cite{simon2006optimal} (Chapter 5.1)
            \State $\hat{f}(n) \leftarrow \hat{X}_0(n)$
    		\State // Extrema Detecting
            \If { $\text{abs}(\hat{X}_1(n)) < \epsilon$  and $\hat{X}_2(n) > 0$}
                \State $\mathbbm{E}_m \leftarrow \{n\} \cup \mathbbm{E}_m$ ~~~~// Minimum reached
            \ElsIf { $\text{abs}(\hat{X}_1(n)) < \epsilon$  and $\hat{X}_2(n) < 0$}
                \State $\mathbbm{E}^m \leftarrow \{n\} \cup \mathbbm{E}^m$ ~~~~// Maximum reached
            \EndIf
            \\
    		\If {end of getting $x(n)$}
    		  \State Break while
    		\EndIf
		\EndWhile
        \Ensure estimated mean $\hat{f}(n)$; minima set $\mathbbm{E}_m$; maxima set $\mathbbm{E}^m$
	\end{algorithmic}
\end{algorithm}

\begin{remark}\label{rem:TVLAP-and-others-time-series}
The Level model and Holt's method (also known as Linear Trend model) mentioned in \cite{durbin2012time} are special cases of TVLAP. When $K = 0$, TVLAP becomes the recursive-form Level model. If $K = 1$, TVLAP degenerates into the recursive-form Holt's method.
\end{remark}

\begin{remark}\label{rem:TVLAP-and-others-target-tracking}
In target tracking community \cite{li2003survey}, one branch of signal processing problem, the canonical Static model, Constant Velocity (CV) model, and Constant Acceleration (CA) model are special cases of TVLAP. When $K = 0$, TVLAP gives the Static model. If $K = 1$, we have the CV model. If $K = 2$, TVLAP degenerates into the CA model.
\end{remark}
%\begin{remark}\label{rem:risk-bound}
%  Note that the {\textbf{\textit{RISK BOUND}}} of the Algorithm \ref{algo:Mean-Extrema-Estimating} is revealed by the covariance matrix $\bm P$. This is easy to understand from the theory of Kalman filter \cite{simon2006optimal}.
%\end{remark}

\subsection{Risk Bound of TVLAP-KF for Estimation and Forecasting}
Note that the risk bound of TVLAP-KF for estimating and forecasting is given by estimation error covariance ($\bm P_{n|n}$, $n$ is current discrete time index) and forecasting error covariance ($\bm P_{n+k|n}$, $k$ is future forecasting steps) returned by Kalman filter.

\subsection{Reliability Guarantee of TVLAP-KF}
In this section, we are concerned to analyze the performances of the proposed TVLAP-KF. That is, we need to investigate whether the TVLAP-KF could recursively approximate $p(t)$ and its derivatives defined in (\ref{eq:taylor-expansion-discrete-time}) with satisfying accuracy.

Before we start, we first give two definitions regrading the observability and contractility of a linear time-invariant system.
\begin{definition}\label{def:observability}
  The linear time-invariant system defined as (\ref{eq:linear-system-x}) is uniformly completely observable if the matrix $\bm O$ defined by the matrices pair $[\bm \Phi, \bm H]$:
  \begin{equation}
    \bm O = [\bm H', \bm \Phi' \bm H', ..., (\bm \Phi')^{K} \bm H']',
  \end{equation}
  is of full rank.
\end{definition}

\begin{definition}\label{def:controbility}
  The linear time-invariant system defined as (\ref{eq:linear-system-x}) is uniformly completely controllable if the matrix $\bm C$ defined by the matrices pair $[\bm \Phi, \bm G]$:
  \begin{equation}
    \bm C = [\bm G, \bm \Phi \bm G, ..., \bm \Phi^{K} \bm G],
  \end{equation}
  is of full rank.
\end{definition}

As we can see, in order to calculate the rank of the matrices $\bm O$ and $\bm C$, we must cope with the calculation of the power of the matrix $\bm \Phi$, specifically, $\bm \Phi^K (T)$. According to the speciality of our defined $\bm \Phi$, we have Lemma \ref{lmm:decomposition-Phi}.

\begin{lemma}\label{lmm:decomposition-Phi}
%\begin{equation}
  $\bm \Phi^K (T) = \bm \Phi (KT)$.
%\end{equation}
\end{lemma}
\begin{proof}
Actually, there exists a matrix
\begin{equation}
    \bm A =  \left[
                    \begin{array}{*{45}{c}}
                        {0}&{1}&{0}&{\cdots}&{0}&{0}\\
                        {0}&{0}&{1}&{\cdots}&{0}&{0}\\
                        %{0}&{0}&{0}&{\cdots}&{0}&{0}\\
                        {\vdots}&{\vdots}&{\vdots}&{\ddots}&{\vdots}&{\vdots}\\
                        {0}&{0}&{0}&{\cdots}&{0}&{1}\\
                        {0}&{0}&{0}&{\cdots}&{0}&{0}
                    \end{array}
                \right],
\end{equation}
such that
%\begin{equation}\label{eq:decomposition-Phi}
    $\bm \Phi(T) = e^{\bm AT}$.
%\end{equation}
Thus, $\bm \Phi^K (T) = e^{K\bm AT} = \bm \Phi (KT)$. That is,
\begin{equation}
    \bm \Phi^K(T) =  \left[
                    \begin{array}{*{45}{c}}
                        {1}&{KT}&{\frac{(KT)^2}{2}}&{\cdots}&{\frac{(KT)^K}{K!}}\\
                        {0}&{1}&{KT}&{\cdots}&{\frac{(KT)^{K-1}}{(K-1)!}}\\
                        {0}&{0}&{1}&{\cdots}&{\frac{(KT)^{K-2}}{(K-2)!}}\\
                        {\vdots}&{\vdots}&{\vdots}&{\ddots}&{\vdots}\\
                        {0}&{0}&{0}&{\cdots}&{1}
                    \end{array}
                \right].
\end{equation}
%Note that not all the transition matrix $\bm \Phi$ of a general linear time-invariant system can be represented as (\ref{eq:decomposition-Phi}).
\end{proof}

\begin{lemma}\label{lmm:vandermonde-matrix}
The Vandermonde matrix defined as
\begin{equation}
    V=\left[\begin{array}{*{45}{c}}
    1 & \alpha_1 & \alpha_1^2 & \dots & \alpha_1^{n-1}\\
    1 & \alpha_2 & \alpha_2^2 & \dots & \alpha_2^{n-1}\\
    1 & \alpha_3 & \alpha_3^2 & \dots & \alpha_3^{n-1}\\
    \vdots & \vdots & \vdots & \ddots &\vdots \\
    1 & \alpha_m & \alpha_m^2 & \dots & \alpha_m^{n-1}
    \end{array}\right],
\end{equation}
is of full rank if $\forall i \ne j$, we have $\alpha_j \ne \alpha_i$.
\end{lemma}
\begin{proof}
Since the determinant of $V$ is $\det(V) = \prod_{1 \le i < j \le n} (\alpha_j - \alpha_i)$ (see \cite{roger1994topics}, chapter 6.1), the lemma stands.
\end{proof}

\begin{lemma}\label{lmm:observability}
  The linear time-invariant system defined in (\ref{eq:linear-system-x}) is uniformly completely observable, if $K$ is not very large.
\end{lemma}
\begin{proof}
  \begin{equation}
    %\begin{array}{cl}
    \bm O = \left[
      \begin{array}{c}
        \bm H \\
        \bm H \bm \Phi \\
        \vdots \\
        \bm H \bm \Phi^{K}
      \end{array}
      \right]
          = \left[
      \begin{array}{*{45}{c}}
             {1}&{0}     &{0}               &{\cdots}   &{0}\\
             {1}&{T}     &{\frac{(T)^2}{2}} &{\cdots}   &{\frac{(T)^K}{K!}} \\
             {1}&{2T}     &{\frac{(2T)^2}{2}} &{\cdots}   &{\frac{(2T)^K}{K!}} \\
        {\vdots}&{\vdots}&{\vdots}          &{\ddots}   &{\vdots}\\
             {1}&{KT}    &{\frac{(KT)^2}{2}}&{\cdots}   &{\frac{(KT)^K}{K!}}
      \end{array}
      \right].
    %\end{array}
  \end{equation}

  Note that, if $K$ is very large, many entries of $\bm O$ would tend to zeroes. Thus, if $K$ is not very large, by Lemma \ref{lmm:vandermonde-matrix}, we have
  \begin{equation}
    \begin{array}{cl}
    rank(\bm O) = rank\left(\left[
      \begin{array}{*{45}{c}}
             {0^0}&{0^1}     &{0^2}       &{\cdots}   &{0^K}\\
             {1^0}&{1^1}     &{1^2}       &{\cdots}   &{1^K} \\
             {2^0}&{2^1}     &{2^2}     &{\cdots}   &{2^K} \\
        {\vdots}&{\vdots}&{\vdots}  &{\ddots}   &{\vdots}\\
             {K^0}&{K^1}     &{K^2}     &{\cdots}   &{K^K}
      \end{array}
      \right]\right)
      = K+1,
    \end{array}
  \end{equation}
  meaning $\bm O$ is of full rank. According to Definition \ref{def:observability}, this lemma stands.
\end{proof}

\begin{lemma}\label{lmm:controbility}
  The linear time-invariant system defined in (\ref{eq:linear-system-x}) is uniformly completely controllable, if $K$ is not very large and $\bm G$ is given as one of the following cases:
  \begin{enumerate}[(a)]
    \item $\bm G_1 = [\frac{T^K}{K!},...,T,1]'$;
    \item $\bm G_2 = diag\{\frac{T^K}{K!},...,T,1\}$;
    \item $\bm G_3$ as an identity matrix $\bm I$. $\bm I$ denotes the identity matrix with proper dimensions.
  \end{enumerate}
\end{lemma}
\begin{proof}
Let $\bm C_{\bm \Phi, \bm G}$ denotes the controllability matrix defined by the pair $[\bm \Phi, \bm G]$. Since
  %\begin{equation}
    $\bm C = [\bm G, \bm \Phi \bm G, ..., \bm \Phi^{K} \bm G]$,
  %\end{equation}
$rank(\bm C_{\bm \Phi, \bm G_3}) = K+1$ (full rank) is easy to check. Due to $rank(\bm C_{\bm \Phi, \bm G_2}) = rank(\bm C_{\bm \Phi, \bm G_3})$, $rank(\bm C_{\bm \Phi, \bm G_2}) = K+1$ also holds. As for $\bm C_{\bm \Phi, \bm G_1}$, we have
%\begin{figure*}[htbp]
%\centering
  \begin{equation}\label{eq:controbility-check-G1}
    \begin{array}{cl}
    \bm C_{\bm \Phi, \bm G_1} &= \left[
        \bm G_1,
        \bm \Phi \bm G_1,
        \cdots,
        \bm \Phi^{K}  \bm G_1
      \right] \\
          &= \left[
      \begin{array}{*{45}{c}}
             {\displaystyle \frac{T^K}{K!}}           &{\displaystyle \sum_{i = 0}^{K} \frac{(1T)^i}{i!} \frac{(T)^{K-i}}{(K-i)!} }
                                        %&{\displaystyle \sum_{i = 0}^{K} \frac{(2T)^i}{i!} \frac{(T)^{K-i}}{(K-i)!} }
                                        &{\cdots}
                                        &{\displaystyle \sum_{i = 0}^{K} \frac{(KT)^i}{i!} \frac{(T)^{K-i}}{(K-i)!} } \\

             {\displaystyle \frac{T^{K-1}}{{K-1}!}}   &{\displaystyle \sum_{i = 0}^{K-1} \frac{(1T)^i}{i!} \frac{(T)^{K-1-i}}{(K-1-i)!}}
                                        %&{\displaystyle \sum_{i = 0}^{K-1} \frac{(2T)^i}{i!} \frac{(T)^{K-1-i}}{(K-1-i)!}}
                                        &{\cdots}
                                        &{\displaystyle \sum_{i = 0}^{K-1} \frac{(KT)^i}{i!} \frac{(T)^{K-1-i}}{(K-1-i)!}} \\

             {\vdots}                   &{\vdots}
                                        %&{\vdots}
                                        &{\ddots}
                                        &{\vdots}\\
             {T}                        &{\displaystyle \sum_{i = 0}^{1} \frac{(1T)^i}{i!} \frac{(T)^{1-i}}{(1-i)!} }
                                        %&{\displaystyle \sum_{i = 0}^{1} \frac{(2T)^i}{i!} \frac{(T)^{1-i}}{(1-i)!} }
                                        &{\cdots}
                                        &{\displaystyle \sum_{i = 0}^{1} \frac{(KT)^i}{i!} \frac{(T)^{1-i}}{(1-i)!} } \\

             {1}                        &{1}
                                        %&{1}
                                        &{\cdots}
                                        &{1}
      \end{array}
      \right].
    \end{array}
  \end{equation}

  %\noindent\rule{\textwidth}{.5pt}%\vskip3pt
%\end{figure*}
By binomial theorem, the entry of $\bm C_{\bm \Phi, \bm G_1}$ at $(I, J)$ is therefore
\begin{equation}
\begin{array}{cll}
\bm C_{\bm \Phi, \bm G_1} (I, J)    &= \displaystyle \sum_{i = 0}^{K - I} \frac{(JT)^i T^{K-I-i}}{i! (K-I-i)!} \\
                &= \displaystyle \frac{1}{(K-I)!} \left(JT+T\right)^{K-I},
\end{array}
\end{equation}
where $I, J = 0,1,2,...,K$,
giving $\bm C_{\bm \Phi, \bm G_1}$ further as
  \begin{equation}\label{eq:controbility-check-G1-simplified}
    \bm C_{\bm \Phi, \bm G_1}   = \left[
      \begin{array}{*{45}{c}}
             {\frac{T^K}{K!}}           &{\frac{(2T)^K}{K!}}
                                        &{\frac{(3T)^K}{K!}}
                                        &{\cdots}
                                        &{\frac{[(K+1)T]^K}{K!}} \\

             {\frac{T^{K-1}}{{K-1}!}}   &{\frac{(2T)^{K-1}}{{K-1}!}}
                                        &{\frac{(3T)^{K-1}}{{K-1}!}}
                                        &{\cdots}
                                        &{\frac{[(K+1)T]^{K-1}}{{K-1}!}}\\

             {\vdots}                   &{\vdots}           &{\vdots}          &{\ddots}   &{\vdots}\\
             {T}                        &2T
                                        &3T
                                        &{\cdots}
                                        &(K+1)T \\

             {1}                        &{1}
                                        &{1}
                                        &{\cdots}
                                        &{1}
      \end{array}
      \right].
  \end{equation}

Note that, if $K$ is very large, many entries of $\bm C_{\bm \Phi, \bm G_1}$ would tend to zeroes. Thus, if $K$ is not very large, by Lemma \ref{lmm:vandermonde-matrix}, we have
  \begin{equation}
    \begin{array}{lll}
        rank(\bm C_{\bm \Phi, \bm G_1})  &= rank\left(\left[
          \begin{array}{*{45}{c}}
                 {1^K}                      &{2^K}
                                            &{\cdots}
                                            &{(K+1)^K} \\

                 {1^{K+1}}                  &{2^{K-1}}
                                            &{\cdots}
                                            &{(K+1)^{K-1}}\\

                 {\vdots}                   &{\vdots} &{\ddots}   &{\vdots}\\

                 {1^1}                      &2^1
                                            &{\cdots}
                                            &(K+1)^1 \\

                 {1^0}                      &{2^0}
                                            &{\cdots}
                                            &{(K+1)^0}
          \end{array}
          \right]\right)
          & = K+ 1.
    \end{array}
  \end{equation}
Since $\bm C_{\bm \Phi, \bm G_1}$ defined in (\ref{eq:controbility-check-G1}) is rank-sufficiency, this lemma stands.
\end{proof}

Now, it is sufficient to give the theorem below to guarantee the reliability of our TVLAP-KF.
\begin{theorem}\label{thm:convergence}
  For any given norm-finite $\bmh X_{0|0}$, if $\bm \Phi$, $\bm G$ and $\bm R$ are bounded, $[\bm \Phi, \bm H]$ is uniformly completely observable, and  $[\bm \Phi, \bm G]$ is uniformly completely controllable, then
  \begin{equation}\label{eq:convergence-X}
    \bmh X_{n|n} \to_d \bm X_{n}, \text{as } n \to \infty,
  \end{equation}
  meaning
  \begin{equation}\label{eq:convergence-p}
    \hat p^{(k)}(n) \to_d p^{(k)}(n), \text{as } n \to \infty, \forall k = 0,1,2,...,K.
  \end{equation}
  %where $\bm R^{-1/2} {\bm R^{-1/2}}’ = \bm R$ and $\bm Q^{-1/2} {\bm Q^{-1/2}}’ = \bm Q$.
\end{theorem}
\begin{remark}
  Note that in Theorem \ref{thm:convergence}, the notation $\to_d$ means convergence in distribution, for example, $\bmh X_{n|n} \to_d \bm X_{n}$ admits $\left[\bmh X_{n|n} - \bm X_{n}\right] \to_d \bm N(\bm 0, \bm P_{k|k})$ where $\bmh X_{n|n}$ means the \textit{a posterior} estimation of $\bm X_n$ given by Kalman filter; $\bm N(\cdot, \cdot)$ means a  multivariate normal distribution; and $\bm P_{k|k}$ is the \textit{a posterior} estimation covariance returned by Kalman filter.
\end{remark}
\begin{proof}
  According to \cite{kalman1961new,anderson1971stability,anderson2012optimal} (see chapter 4.4 of \cite{anderson2012optimal}), with support of our Lemma \ref{lmm:observability} and Lemma \ref{lmm:controbility}, this theorem holds. Note that
  %\begin{equation}
    $rank(\bm O_{\bm \Phi, \bm H}) = rank(\bm O_{\bm \Phi, \bm H \bm R^{-1/2}})$,
  %\end{equation}
  and
  %\begin{equation}
    $rank(\bm C_{\bm \Phi, \bm G}) = rank(\bm C_{\bm \Phi, \bm G \bm Q^{-1/2}})$,
  %\end{equation}
  where $\bm R^{-1/2} (\bm R^{-1/2})’ = \bm R$ and $\bm Q^{-1/2} (\bm Q^{-1/2})’ = \bm Q$. Since $\bm R$ and $\bm Q$ are positive definite, the decomposition is possible. In above, $\bm O_{\bm \Phi, \bm H}$ denotes the observability matrix defined by the pair $[\bm \Phi, \bm H]$. The notation conventions keep same to $\bm C_{\bm \Phi, \bm G}$, $\bm O_{\bm \Phi, \bm H \bm R^{-1/2}}$ and $\bm C_{\bm \Phi, \bm G \bm Q^{-1/2} }$.
\end{proof}

\subsection{Issue of Selecting the Model Order $K$ and the Time Gap $T$}\label{subsec:selecting-para}
It is easy to see that the core of the TVLAP model is the matrix $\bm \Phi$ defined in (\ref{eq:system-matrix}). It relates to the parameters $K$ and $T$, and the model performances depend much on the proper values of them.
\begin{itemize}
  \item {\textbf{Choosing $T$}.} For a typical time series, $T = 1$ in theory. However, in practice, the suggested value of $T$ should be $T \le 1$. This is because the original series $x(n)$ contains the noise (high-frequency) component, meaning the estimation error to derivatives would never be zeroes even though the true values are zeroes. Therefore, if we have $T<1$, the impact introduced by the estimation errors to high-order derivatives could be weaken or eliminated. This is because the term $T^k/k!$ will rapidly converge to zero if $T < 1$ as $k$ increases.
  \item {\textbf{Choosing $K$}.} In theory, for $K$, the larger, the better. However, in practice, due to the exists of noise and Runge phenomenon in polynomial fitting, $K$ should not be extremely large. The suggested value of $K$ should be $2 \sim 8$. $K = 4$ is a typical option. This is the experience of authors obtained in simulation studies. On the other hand, as stated in Lemma \ref{lmm:observability} and Lemma \ref{lmm:controbility}, the selected $K$ should not be very large so that the observability and controllability matrices are non-singular.
\end{itemize}

\begin{remark}
It is possible that $T$ may have explicit meaning in practice for a general time series, meaning we cannot assign value to it arbitrarily. For example, if the time series is the sales volume of apple in a store per day. The meaning of $T$ should be $1$ with unit of $Day$. It seems inconsistent to our suggestion of $T \le 1$. However, fortunately this is not an issue because $f^{(k)}$ is a variable, meaning the difference could be compensated by the estimate to $f^{(k)}$. Plus, the estimates to $f^{(k)}$ may become larger than their real values, advantaging that they become easier to observe, if we do not care about the exactly true values of $f^{(k)}$.
\end{remark}

\subsection{Modeling Error Between $p(n)$ and $f(n)$}
Although Theorem \ref{thm:convergence} asserts the sufficiency of the proposed TVLAP-KF, we should mention here that our TVLAP-KF may still suffer from some numerical errors. This is because the real trend of $x(n)$ is $f(n)$, not $p(n)$. Facing this problem of modeling error and/or unknown disturbances, the adaptive Kalman filter with unknown disturbances \cite{kim2000robust} and/or unknown inputs \cite{yong2016unified} are developed to bridge (or eliminate) the modeling gap. Specifically, instead of studying (\ref{eq:linear-system-x}), we should focus on
\begin{equation}\label{eq:linear-system-x-unknown-inputs}
    \left\{ \begin{array}{rll}
        \bm X(n+1) &= \bm \Phi \bm X(n) + \bm M \bm A(n) +  \bm G \bm W(n) \\
        \bm Y(n) &= \bm H \bm X(n) + \bm V(n),
    \end{array} \right.
\end{equation}
where $\bm A(n)$ is the unknown input driven by a known matrix $\bm M$. The effort here should be on estimating $\bm X(n)$ in presence of the unknown $\bm A(n)$. Obviously, in (\ref{eq:linear-system-x-unknown-inputs}), $\bm A(n)$ is used to compensate or eliminate the modeling bias (modeling error between $f(n)$ and $p(n)$). Thus in this paper, the TVLAP-KF does not merely refer to the canonical Kalman filter, also includes any proper members in Kalman filter family.
\begin{remark}
We consider this modeling issue just for theoretical completeness. In engineering, according to authors' experience, it does not necessarily consider this issue. This is because $f(n)$ is usually guaranteed to be continuous. Eq. (\ref{eq:Weierstrass-approximation}) could ensure the estimation error due to modeling bias (between $p(n)$ and $f(n)$) to be acceptable for a real problem. Besides, more complicated algorithm would introduce extra computation burden which is undesired for online signal processing problems.
\end{remark}

\subsection{General Methodology for Non-White Noise}\label{sec:holistic-methodology}
Unfortunately, the Kalman filter is optimal only for white noise. However, the Kalman filter is still to some degree effective to return a feasible solution, in practice. The good news is that there exists the exact method for colored noise Kalman filter. We in the following derive the TVLAP-KF model for a time series with colored (Non-White) noise. That is, we no longer assume $x_s(n)$ to be $white(n)$. Instead, we investigate the general colored case of it.

Suppose the remainder (i.e., the noise part) $\hat{x}_s(n)$ could be modelled by $ARMA(p,q|\bm \varphi, \bm \theta)$ with the transfer function as
\begin{equation}\label{eq:regular-ARMA}
    H(z) = \frac{\theta_0 + \theta_1 z^{-1} + ... + \theta_q z^{-q}}{1 + \varphi_1 z^{-1} + ... + \varphi_p z^{-p}}.
\end{equation}
It means the input of this ARMA system is a 1-dimensional Gaussian white sequence $\bm \varepsilon(n)$ and the output is ${x}_s(n) =: \bm V(n) = V(n)$. Note that the first coefficient of the denominator polynomial is normalized to $1$. Note also that the white noise case is the special case of $ARMA(p,q)$ with $H(z) = \theta_0$, that is, $ARMA(0,0)$.

Let $r := max\{p, q\}$, $\varphi_j := 0,~\forall j > p$, and $\theta_j := 0,~\forall j > q$. Then we have an alternative representation of (\ref{eq:regular-ARMA}) as
\begin{equation}\label{eq:regular-ARMA-adjusted}
    \begin{array}{rlll}
      H(z)  &= \displaystyle \frac{\theta_0 + \theta_1 z^{-1} + ... + \theta_{r} z^{-r}}{1 + \varphi_1 z^{-1} + ... + \varphi_r z^{-r}} \\
            &~ \\
            &= \displaystyle \frac{\theta_0 z^{r} + \theta_1 z^{r-1} + ... + \theta_{r}}{z^{r} + \varphi_1 z^{r-1} + ... + \varphi_r} \\
            &~ \\
            &= \theta_0 + \displaystyle \frac{\beta_1 z^{r-1} + ... + \beta_{r}}{z^{r} + \varphi_1 z^{r-1} + ... + \varphi_r},
    \end{array}
\end{equation}
where $\beta_i := \theta_i - \theta_0 \varphi_i,~i=1,2,...,r$.

Therefore, the state space counterpart of (\ref{eq:regular-ARMA}) is
\begin{equation}\label{eq:state-space-ARMA}
\left\{
    \begin{array}{rlll}
        \bm \xi(n+1) &= \bm \Xi \bm \xi(n) + \bm \Upsilon \bm \varepsilon(n) \\
        \bm V(n)     &= \bm \Pi \bm \xi(n) + \bm \Lambda \bm \varepsilon(n)
    \end{array}
\right.
\end{equation}
where
\begin{equation}\label{eq:system-matrix-ARMA}
        \bm \Xi =  \left[
                        \begin{array}{*{45}{c}}
                            {0}&{1}&{0}&{\cdots}&{0}\\
                            {0}&{0}&{1}&{\cdots}&{0}\\
                            {\vdots}&{\vdots}&{\vdots}&{\ddots}&{\vdots}\\
                            {0}&{0}&{0}&{\cdots}&{1}\\
                            {-\varphi_{r}}&{-\varphi_{r-1}}&{-\varphi_{r-2}}&{\cdots}&{-\varphi_1}
                        \end{array}
                    \right],
\end{equation}
\begin{equation}\label{eq:driven-matrix-ARMA}
    \bm \Upsilon =  \left[
                    \begin{array}{*{45}{c}}
                        {0}\\
                        {0}\\
                        {\vdots}\\
                        {0}\\
                        {1}
                    \end{array}
                \right],
\end{equation}
\begin{equation}\label{eq:obsv-matrix-ARMA}
  \bm \Pi  = \left[ \begin{array}{*{45}{c}} \beta_r & \beta_{r-1} & \cdots & \beta_2 & \beta_1  \end{array} \right],
\end{equation}
and
\begin{equation}\label{eq:obsv-driven-matrix-ARMA}
  \bm \Lambda  = \theta_0.
\end{equation}
Note that this $\bm V$ is conceptually similar to the one in Eq. (\ref{eq:measure-equation}).

Thus, the entire state space model for our general model $x(n) = f(n) + x_s(n)$, namely, Eq. (\ref{eq:general-model-lite}), should be
\begin{equation}\label{eq:linear-system-x-noise-dynamic-included}
    \left\{ \begin{array}{rll}
        \bm X(n+1) &= \bm \Phi \bm X(n) + \bm G \bm W(n) \\
        \bm Y(n) &= \bm H \bm X(n) + \bm V(n) \\
        \bm \xi(n+1) &= \bm \Xi \bm \xi(n) + \bm \Upsilon \bm \varepsilon(n) \\
        \bm V(n)     &= \bm \Pi \bm \xi(n) + \bm \Lambda \bm \varepsilon(n),
    \end{array} \right.
\end{equation}
which, by augmenting the state vector, is equivalent to %Eq. (\ref{eq:linear-system-x-augmented}).
%\begin{figure*}[htbp]
%  \centering
%    \begin{equation}\label{eq:linear-system-x-augmented}
%        \left\{ \begin{array}{rll}
%            \left[ \begin{array}{c}
%              \bm X(n+1)  \\
%              \bm \xi(n+1)
%            \end{array} \right] &=
%            \left[ \begin{array}{*{45}c}
%              \bm \Phi &  \bm 0\\
%              \bm 0 & \bm \Xi
%            \end{array} \right]
%            \left[ \begin{array}{c}
%              \bm X(n)  \\
%              \bm \xi(n)
%            \end{array} \right] +
%            \left[ \begin{array}{*{45}c}
%              \bm G &  \bm 0\\
%              \bm 0 & \bm \Upsilon
%            \end{array} \right]
%            \left[ \begin{array}{c}
%              \bm W(n)  \\
%              \bm \varepsilon(n)
%            \end{array} \right]   \\
%
%
%            \bm Y(n) &=
%            \left[ \begin{array}{*{45}c}
%              \bm H &  \bm \Pi
%            \end{array} \right]
%            \left[ \begin{array}{c}
%              \bm X(n)  \\
%              \bm \xi(n)
%            \end{array} \right] + \bm \Lambda \bm \varepsilon(n)
%        \end{array}, \right.
%    \end{equation}
%\end{figure*}
    \begin{equation}\label{eq:linear-system-x-augmented}
        \left\{ \begin{array}{rll}
            \bmb X(n+1) &= \bmb \Phi \bmb X(n) + \bmb w(n) \\
            \bm Y(n)    &= \bmb H(n) \bmb X(n) + \bmb v(n),
        \end{array} \right.
    \end{equation}
where
\begin{equation}\label{eq:linear-system-x-augmented-state}
    \bmb X(n) :=
        \left[ \begin{array}{c}
          \bm X(n)  \\
          \bm \xi(n)
        \end{array} \right],
\end{equation}
\begin{equation}\label{eq:linear-system-x-augmented-state-matrix}
    \bmb \Phi :=
            \left[ \begin{array}{*{45}c}
              \bm \Phi &  \bm 0\\
              \bm 0 & \bm \Xi
            \end{array} \right],
\end{equation}
\begin{equation}\label{eq:linear-system-x-augmented-noise}
    \bmb w(n) :=
            \left[ \begin{array}{*{45}c}
              \bm G &  \bm 0\\
              \bm 0 & \bm \Upsilon
            \end{array} \right]
            \left[ \begin{array}{c}
              \bm W(n)  \\
              \bm \varepsilon(n)
            \end{array} \right],
\end{equation}
\begin{equation}\label{eq:linear-system-x-augmented-measure-matrix}
    \bmb H :=
            \left[ \begin{array}{*{45}c}
              \bm H &  \bm \Pi
            \end{array} \right],
\end{equation}
and
\begin{equation}\label{eq:linear-system-x-augmented-measure-noise}
    \bmb v(n) := \bm \Lambda \bm \varepsilon(n).
\end{equation}

The system (\ref{eq:linear-system-x-augmented}) could be handled by the Colored Kalman filter (see \cite{simon2006optimal}, chapter 7.1). Note that the covariance matrix between the process noise $\bmb w(n)$ and the measurement noise $\bmb v(n)$ is
\begin{equation}\label{eq:linear-system-x-augmented-colored-cov}
    \begin{array}{rllll}
      \bm E[\bmb w(n) \bmb v^{T}(j)] &:= \bm M(n) \delta_{k-j} =
            \left[ \begin{array}{c}
              \bm 0  \\
              \bm \Upsilon \bmb R \bm \Lambda^{T}
            \end{array} \right] \delta_{n-j}  \\
            &=
            \left[ \begin{array}{c}
              \bm 0  \\
              \bm \Upsilon \bar R \Lambda
            \end{array} \right] \delta_{n-j},
    \end{array}
\end{equation}
where $\delta_{n-j}$ is the Kronecker delta function; $\bmb R = \bar R$ denote the variance of $\bm \varepsilon(n) = \varepsilon(n)$. Now, the last thing to do is to estimate the value of $\bar{R}(n)$.

Eq. (\ref{eq:regular-ARMA}) reveals how $\varepsilon(n)$ generates ${x}_s (n) := V(n)$. Since ${x}_s (n)$ is a WSS process with fixed variance $R(n)$, we can have the variance $\bar{R}(n)$ of $\varepsilon(n)$, according to \cite{diniz2010digital} (see chapter 2.11.1), implicitly defined as
\begin{equation}\label{eq:variance-relation-H}
    \begin{array}{rlll}
        R(n) &= \displaystyle \frac{1}{2\pi} \displaystyle \int_{-\pi}^{\pi} \left| H(e^{jw}) \right|^2 \cdot \bar{R}(n) dw \\
             &= \bar{R}(n) \cdot \displaystyle \frac{1}{2\pi} \displaystyle \int_{-\pi}^{\pi} \left| H(e^{jw}) \right|^2 dw,
    \end{array}
\end{equation}
where $H(e^{jw}) = H(z)|_{z = e^{jw}}$ is the Fourier frequency response of $H(z)$; the term $\left| H(e^{jw}) \right|^2 \cdot \bar{R}(n)$ denotes the power spectra of the output sequence ${x}_s(n)$. Suppose the impulse response of the system $H(z)$ is $h(n)$. According to the Parseval's theorem (see \cite{diniz2010digital}, chapter 2.9.11), we further have
\begin{equation}\label{eq:variance-relation-h}
  R(n) = \bar{R}(n) \cdot \displaystyle \sum_{n=-\infty}^{\infty} h^2(n) = \bar{R}_n \cdot \displaystyle \sum_{n=0}^{\infty} h^2(n),
\end{equation}
namely,
\begin{equation}\label{eq:variance-relation-h2}
  \bar{R}(n) = \displaystyle \frac{R(n)}{\displaystyle \sum_{n=0}^{\infty} h^2(n)}.
\end{equation}
Note that in (\ref{eq:variance-relation-h2}), $R(n)$ has already been estimated from the data series $\hat{x}_s$. For details, see Subsection \ref{subsec:estimate-R}. Note also that when $H(z)$ is stable, $\sum_{n=0}^{\infty} |h(n)|$ is convergent, which means $\sum_{n=0}^{\infty} h^2(n)$ is also convergent. Besides, when $H(z)$ is causal, $h(n) = 0$, $\forall n < 0$. For a real system $H(z)$, the stability and causality are guaranteed.

Reaching here, the TVLAP-KF for a colored-noise time series is wholly built.

\section{Simulation Study}
In this section, we provide some experiments to illustrate possible applications of TVLAP-KF in engineering, including: (a) Online Trend Estimating, Trend Tracking and Extrema Detecting; (b) Online Long-term Prediction; (c) Online Fault Diagnosis of Sensors; and (d) Highly-maneuvering Target Tracking. In order not to limit our framework into one narrow kind of specific problem or application, while ensuring the practical utilization value, we use two kinds of simulation scenarios: (a) Simulated data and scenario for estimation and forecasting problem; (b) Real data collected from real sensors for fault diagnosis.

\subsection{Trend Estimating, Trend Tracking and Extrema Detecting}
Suppose $t = 0:0.1:120$ (thus $T = 0.1$), $x(n) = 5\sin(0.1t) + WG(length(t))$, and let $K = 4$, $\bm Q = diag\{0,0,0,{0.01}^2\}$, $\bm R = 1^2$, where $WG$ means a White Gaussian process with mean of zero and variance of $1$, we then apply the Algorithm \ref{algo:Mean-Extrema-Estimating} to Mean Estimating and Extrema Detecting, and have the simulation results in Fig. \ref{fig:TVLAP-simulation}. Fig. \ref{fig:TVLAP-simulation} as well shows a result of the 200-step prediction. The prediction includes the mean forecasting and extrema forecasting. Note that we desire the low-frequency component of the raw series, since the measure is noised, we should trust more the system equation than the measurements. It is this reason that we have process variance be less than measurement valiance. Note also that the choose of $\bm Q$ and $\bm R$ does not necessarily depend on the true variance of raw series $x(n)$ \cite{li2017approximate}.

Fig. \ref{fig:TVLAP-simulation} (a) and (b) show that we can effectively estimate the $f(n)$ and its derivatives. Fig. \ref{fig:TVLAP-simulation} (b) suggests that when the estimated first-order derivative (green line) is positive, $f(n)$ is increasing; when negative, it is decreasing. Fig. \ref{fig:TVLAP-simulation} (c) shows that when the estimated first-order derivative is zero, $f(n)$ reaches its extrema. Whether they are maxima or minima is based on whether the estimated second-order derivatives at corresponding points are positive or not. For figure clearness, we did not plot the second-order derivative in Fig. \ref{fig:TVLAP-simulation} (c). Readers can redirect to Fig. \ref{fig:TVLAP-simulation} (b) to compare with. The satisfying prediction performances in Fig. \ref{fig:TVLAP-simulation} comes from the fact that the model we designed could preserve the high-order information of the changing pattern.

\begin{figure*}[htbp]
    \centering
    \subfigure[Mean Estimating and Forecasting]{
        \begin{minipage}[htbp]{\linewidth}
            \centering
            \includegraphics[height=4cm]{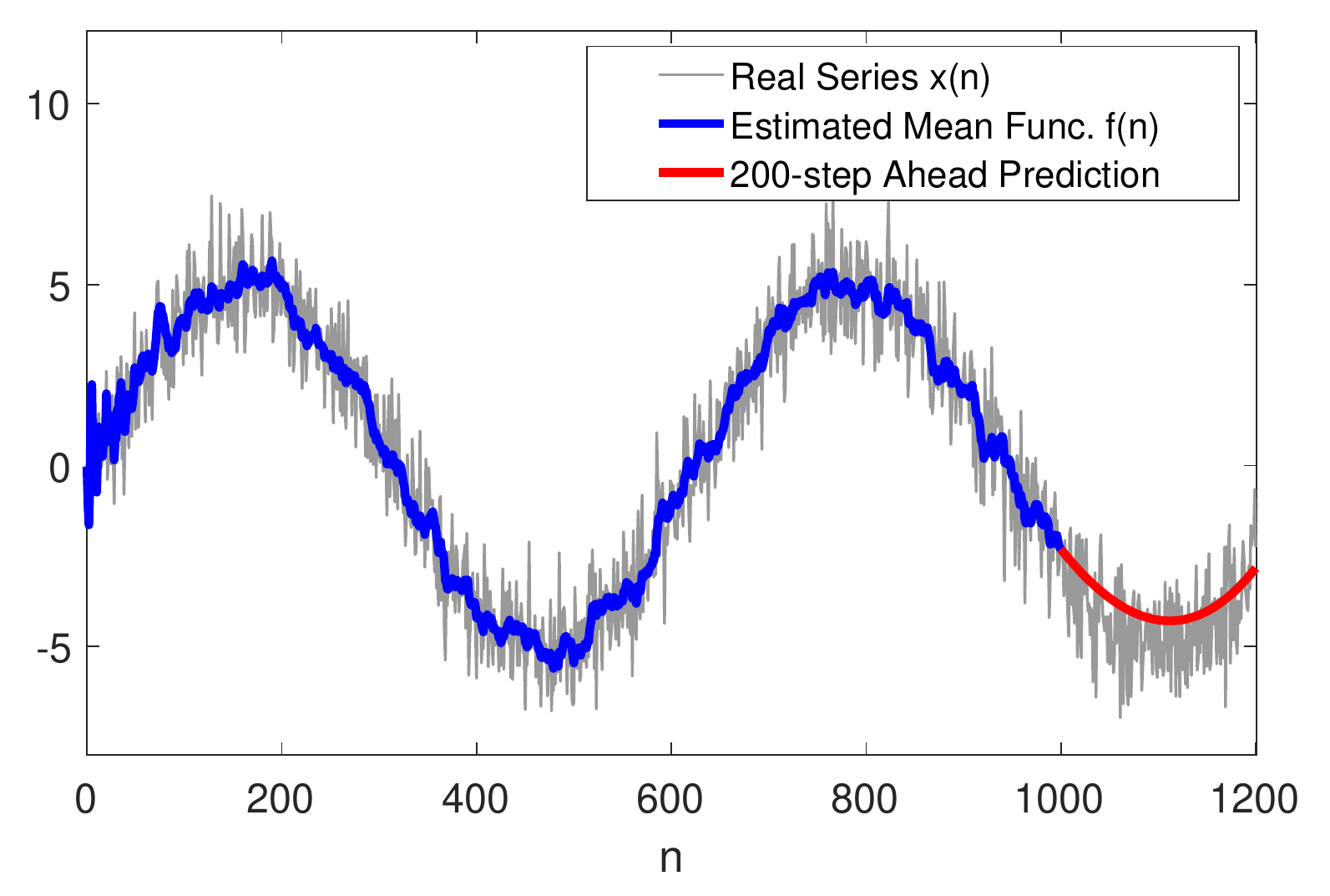}
        \end{minipage}
    }

    \subfigure[Estimated and forecasted derivatives of mean function]{
        \begin{minipage}[htbp]{0.46\linewidth}
            \centering
            \includegraphics[height=4cm]{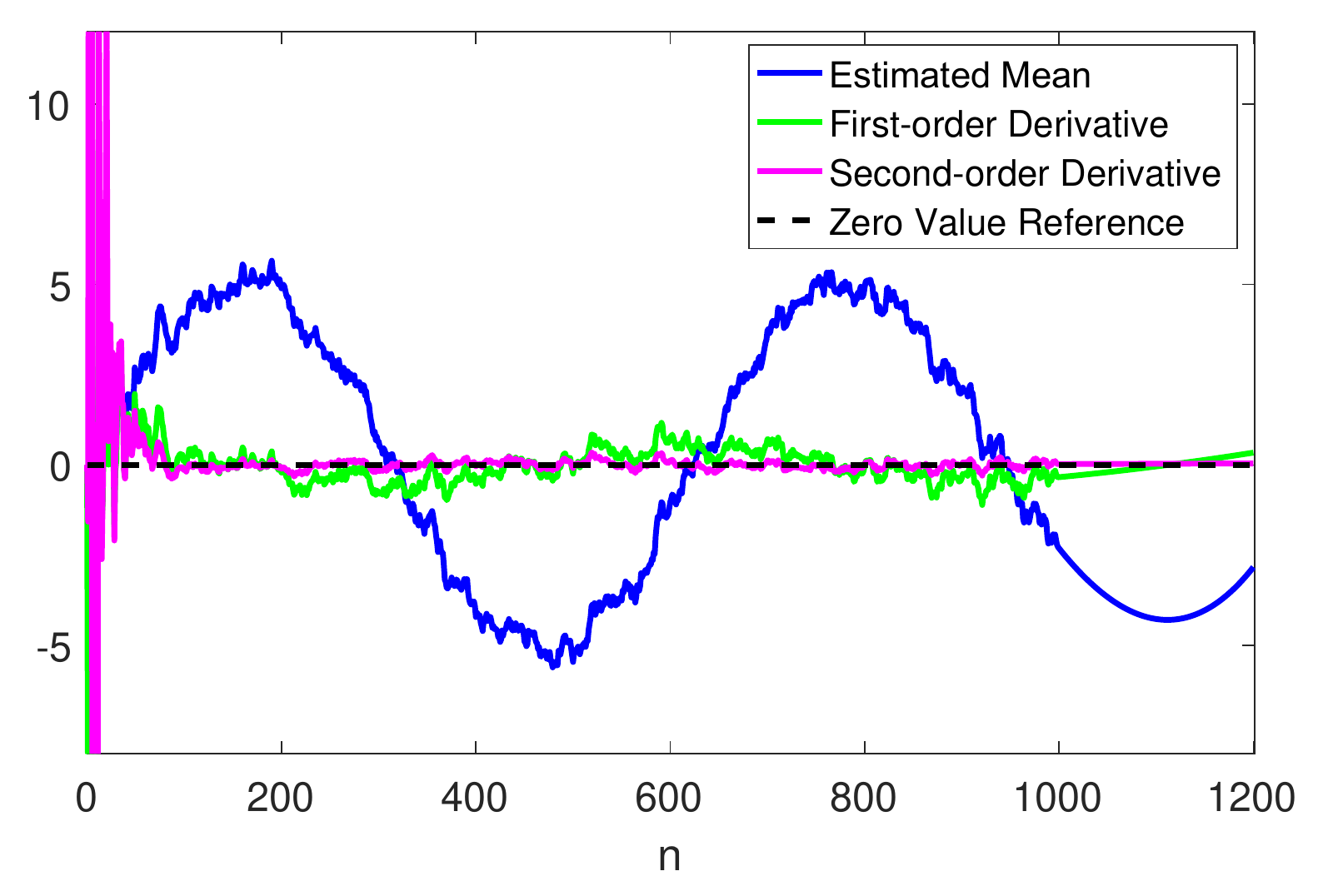}
        \end{minipage}
    }
    \subfigure[Extrema Detecting and Extrema forecasting]{
        \begin{minipage}[htbp]{0.46\linewidth}
            \centering
            \includegraphics[height=4cm]{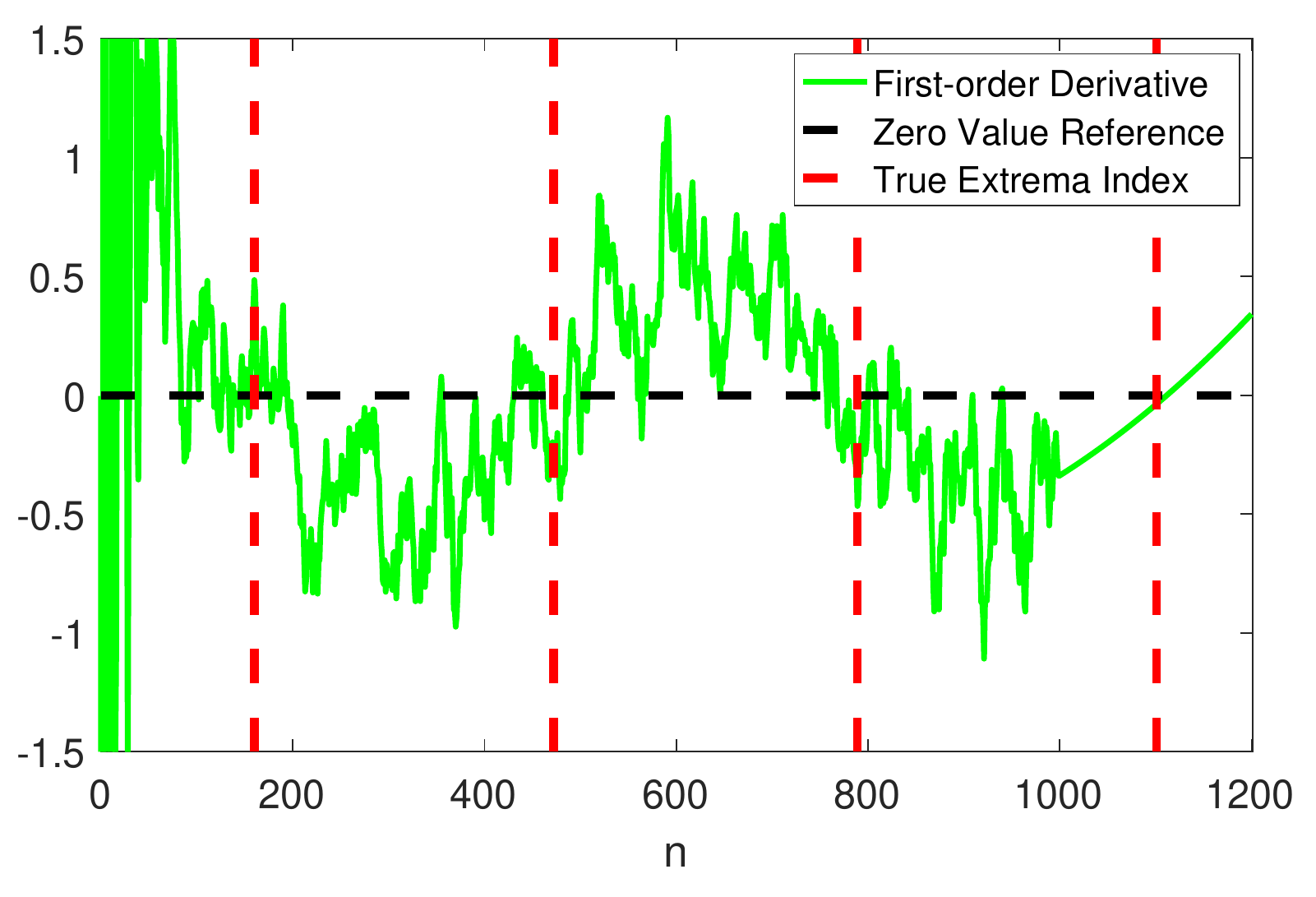}
        \end{minipage}
    }
    \centering
    \caption{Simulation results on trend tracking and extrema detecting of TVLAP model}
    \label{fig:TVLAP-simulation}
\end{figure*}

\subsection{Trend Estimation and Long-term Prediction Performance of TVLAP}\label{subsec:estimation-prediction-comparasion}
In this part, we demonstrate the estimation and prediction performance of TVLAP, with comparison with Holt's method and Local Level method \cite{hyndman2008forecasting,hyndman2018forecasting}. The off-line methods (which are only workable for block data) will not be taken into account. Suppose we have $t = 0:0.1:120$, $x(n) = 5\sin(0.1t) + \exp(0.03t) + WG(length(t))$, the 200-step ahead predictions given by TVLAP ($K = 4$, $\bm Q = diag\{0,0,0,{300}^2\}$, $\bm R = 1^2$, $T = 0.001$), Holt's, and Local Level models are displayed in Fig. \ref{fig:long-term-predction-cmp}.
\begin{figure}[htbp]
    \centering
    \includegraphics[height=4cm]{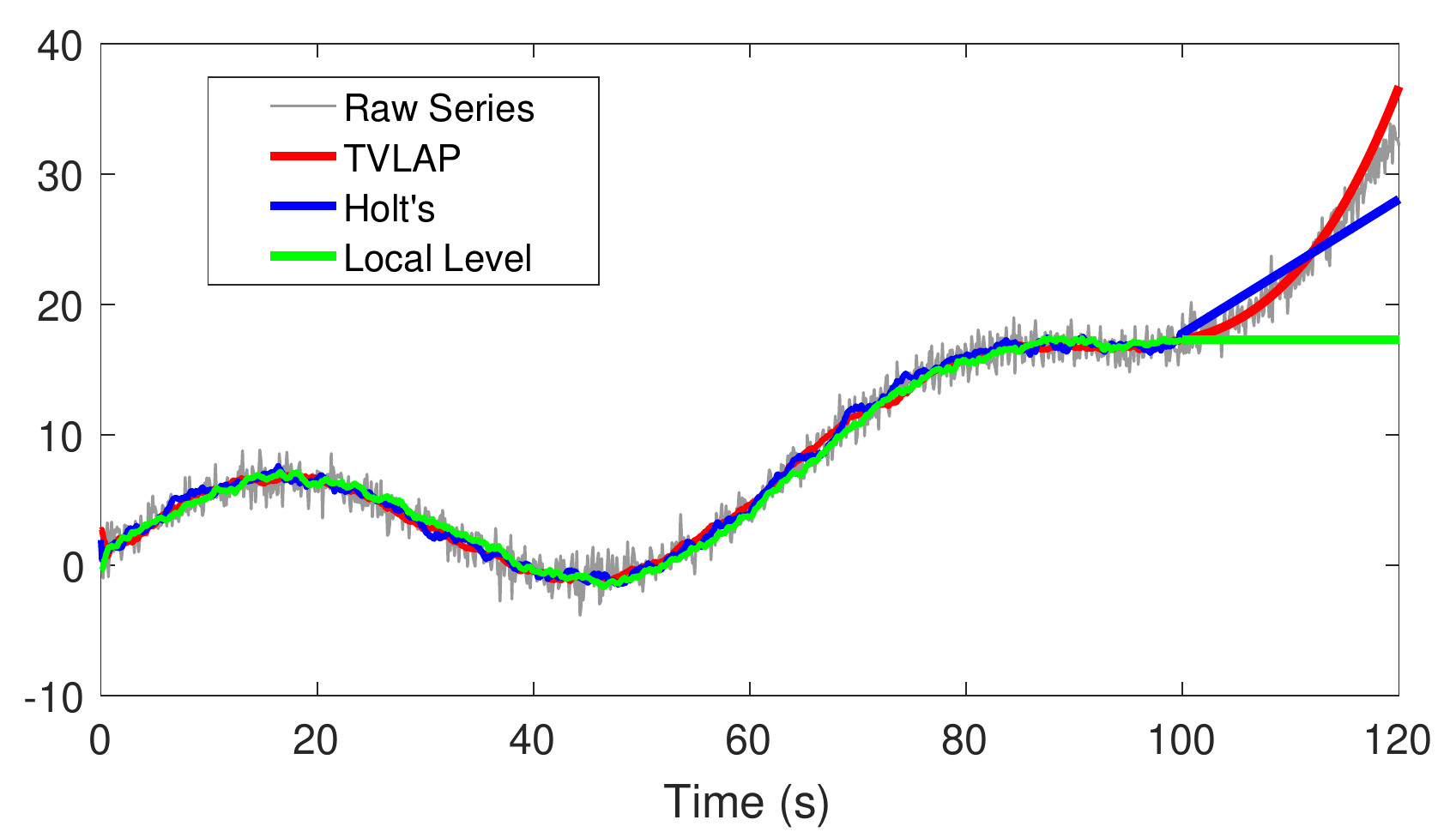}
    \caption{Long-term Prediction Performances of TVLAP, Holt's and Local Level}
    \label{fig:long-term-predction-cmp}
\end{figure}
All the prediction results given in Fig. \ref{fig:long-term-predction-cmp} are the respective best ones among $10$ simulations. The corresponding estimation (time from 0 to 100) and prediction (time from 100 to 120) MSE (mean square error) are given in Table \ref{tab:mse-value}. As we can see, it is deep pattern of data that allows us to make more satisfactory estimation and prediction. This is because, as mentioned earlier (issue of time-delay), the order-insufficient models cannot promptly track the relatively sharp changing pattern of a time series. Therefore, high-order models (with relatively large $K$) are expected.
\begin{table*}[htbp]\label{tab:mse-value}
    \centering
    \caption{Estimation and prediction MSE of TVLAP, Holt's, and Local Level}
    \begin{tabular}{lccccc}
        \toprule
                &  Estimation MSE  & Prediction MSE \\
        \midrule
        TVLAP        & 0.0689   &  2.5979 \\
        \hline
        Holt's        & 0.0790   &  17.5410 \\
        \hline
        Local Level        & 0.1709   &  62.3891 \\
        \bottomrule
    \end{tabular}
    %\hrulefill
\end{table*}

\subsection{Fault Diagnosis of Sensors}
In this part, we show the application of TVLAP-KF in Anchor Selection Problem in indoor positioning based on Ultra-wide Band (UWB) ranging signals \cite{ridolfi2018experimental}. The data in this experiment are real data collected from UWB sensors. In order to improve the positioning performances in an indoor space, for example, a warehouse, we deploy many UWB sensors (anchors) in one space. Due to signal sheltering and complex electromagnetic environment, ranging signals from different anchors may have different ranging performances at different areas. Thus, we aim to select sensors without large error from all the available anchors in one area to localize and track the moving target. The essence of the above issue is actually to diagnose the sensor fault (or detect the anomalies in ranging signals), in online manner. Ranging signals provided by three of all available anchors are showed in Fig. \ref{fig:ranging-signals}.
\begin{figure}[htbp]
    \centering
    \includegraphics[height=4cm]{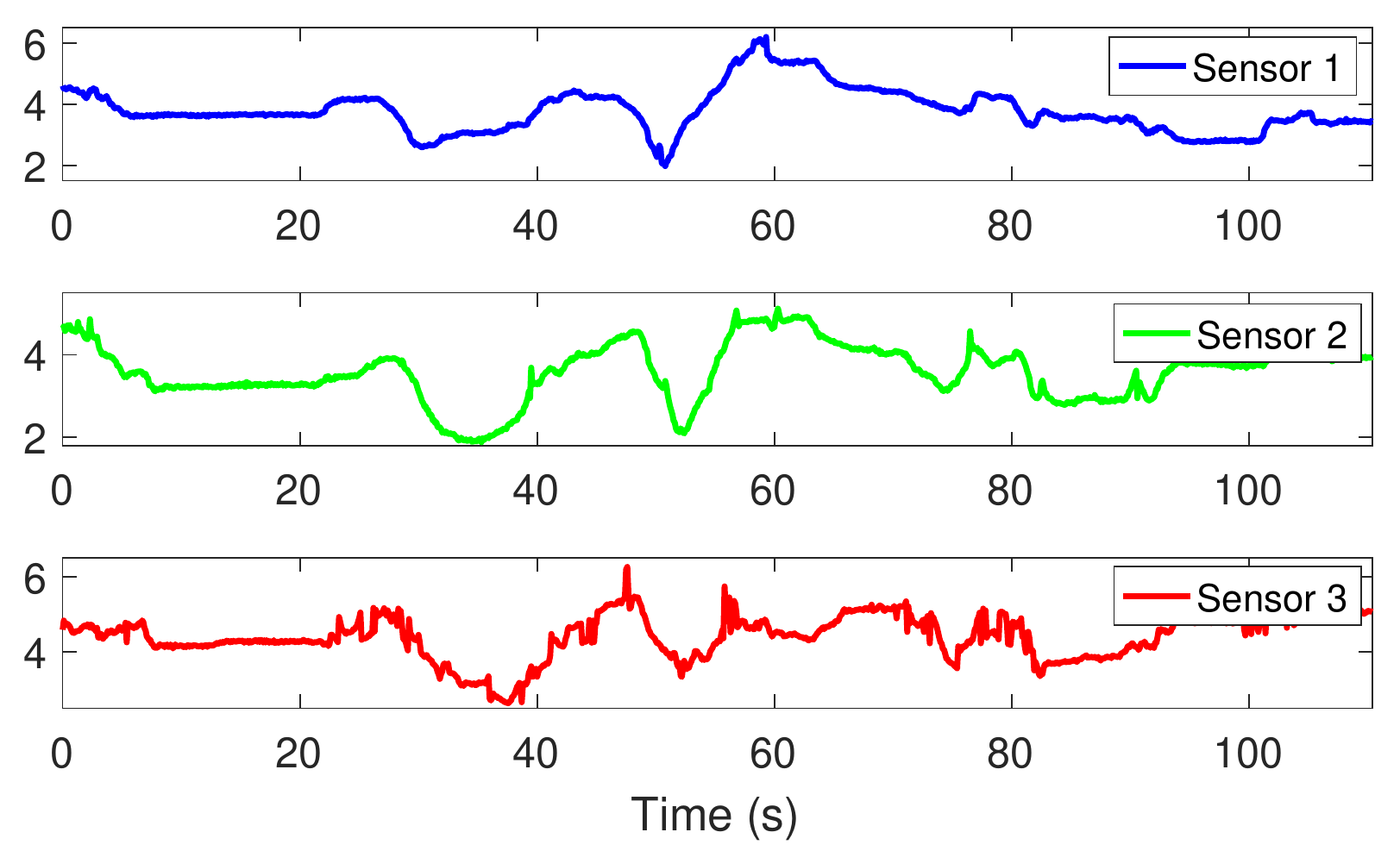}
    \caption{UWB ranging signals from three of anchors}
    \label{fig:ranging-signals}
\end{figure}

Intuitively, the Sensor 3 is with large error, since its ranging signal jumps at many discrete time indices (for instance, when $n = 25 \sim 30$, around $n=48$, and $n = 70 \sim 80$, $n$ is discrete time index). Those jumps are in fact errors because a real moving target cannot maneuver in such a sharp way. On the other hand, if they are indeed generated from sharp maneuvers, Sensor 1 and Sensor 2 should have same jumps in their ranging signals.

We aim to differentiate Sensor 3 from Sensor 1 and Sensor 2 so that Sensor 3 would be excluded to participate in positioning in this area. If we use TVLAP-KF ($K=4$, $R=0.03$, $\bm Q = 500^2$, $T = 0.001$. $R$ is estimated from the real data; $\bm Q$ is set to be large because we in this scenario emphasize more on observations than the system model) to estimate the changing pattern (first-order derivative) of ranging signals, we have Fig. \ref{fig:UWB-sensor-faults}. The variances of the three time series in Fig. \ref{fig:UWB-sensor-faults} are given in Table \ref{tab:sensor-variance-value}.
\begin{figure}[htbp]
    \centering
    \includegraphics[height=4cm]{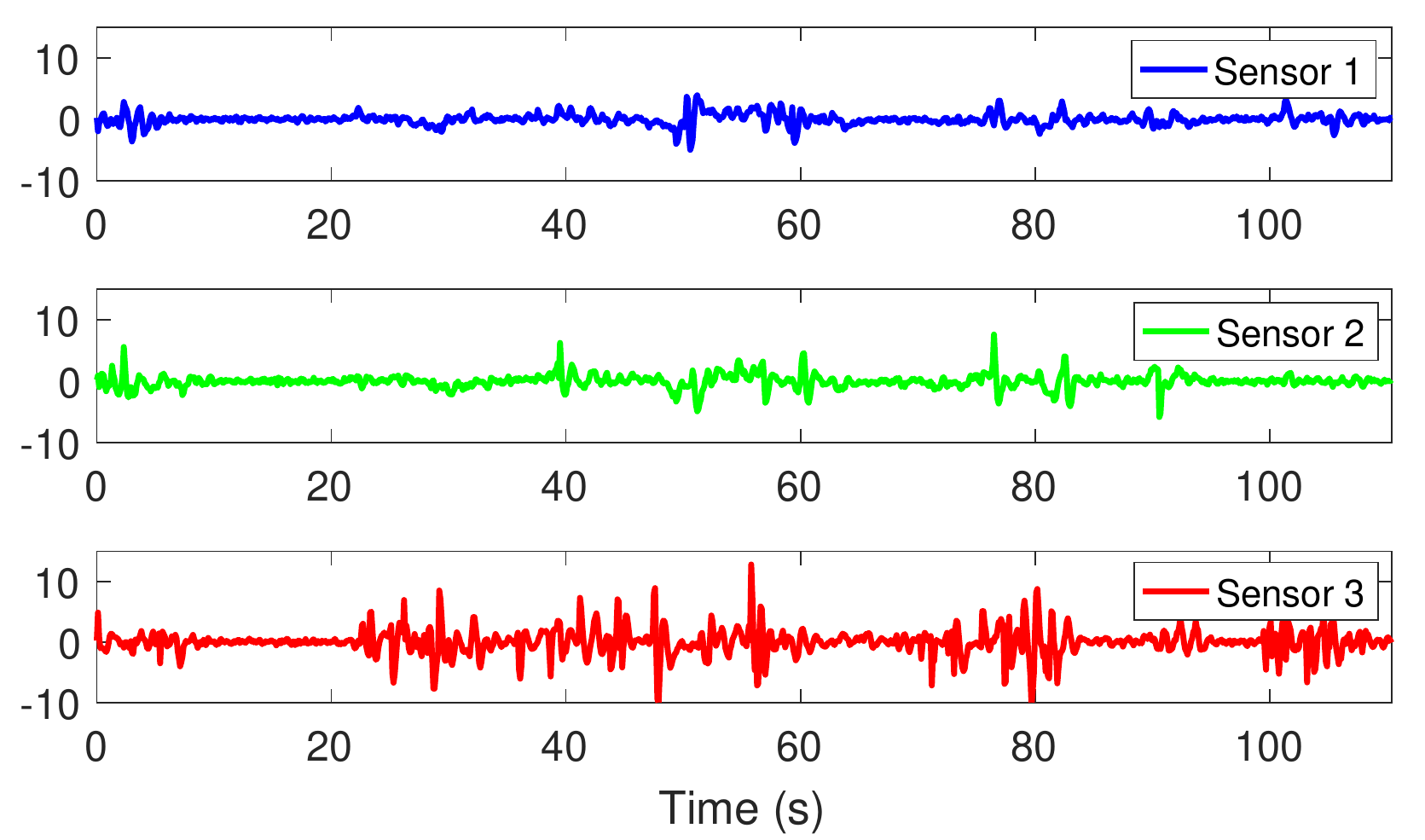}
    \caption{Deep pattern (changing pattern) of three UWB ranging signals}
    \label{fig:UWB-sensor-faults}
\end{figure}

\begin{table*}[htbp]\label{tab:sensor-variance-value}
    \centering
    \caption{Variances of the three time series in Fig. \ref{fig:UWB-sensor-faults}}
    \begin{tabular}{lccccc}
        \toprule
                        &  Sensor 1  & Sensor 2 & Sensor 3 \\
        \midrule
            Variance    &   0.6424   &   0.8293  & 3.3216 \\
        \bottomrule
    \end{tabular}
    %\hrulefill
\end{table*}

From Fig. \ref{fig:UWB-sensor-faults} and Table \ref{tab:sensor-variance-value}, it is easy to tell apart Sensor 3 from Sensor 1 and Sensor 2 because Sensor 3 has significantly large variance (or more outliers) in changing pattern of ranging signal. Note that the variance estimation method of zero-mean sequence $x(n)$ is given as $\sum_{i=1}^{n} x^2(n) /n$ (its online version, namely recursive version, is easy to derive). Note also that it is not reasonable to use time-difference (i.e., $[x(n) - x(n-1)]/T_s$) to estimate the first-order derivative of a noised time series $x(n)$, since the noise would be amplified by difference operator \cite{WANG2019ARIMA}, so that the patten can not be easily observed compared to Fig. \ref{fig:UWB-sensor-faults}. This point is easy to check. Thus the illustration for it is ignored.

\subsection{Highly-maneuvering Target Tracking}
In consideration of necessity, we do not re-design a new simulation for highly-maneuvering target tracking here. Readers are invited to refer to \cite{karsaz2009optimal} for a motivational example, in which the derivative of acceleration (that is, jerk, $K=3$) is considered to improve the tracking performances, compared to the traditional Constant Velocity ($K=1$) model and Constant Acceleration ($K=2$) model.

Essentially, this problem is mathematically similar to the example in Subsection \ref{subsec:estimation-prediction-comparasion}. When a target is highly maneuvering, the order-insufficient models cannot promptly track such sharp manoeuvres. Therefore, high-order models (with relatively large $K$) are expected.

\section{Conclusion}\label{sec:conclusion}
This paper studies the Time-Variant Local Autocorrelated Polynomial model with Kalman filter to investigate the deep pattern of a time series. Its possible applications in engineering are discussed. Simulation suggests that the order-insufficient models cannot promptly track the relatively sharp changes in a time series so that our TVLAP model can outperform other methods both in estimation and in forecasting. Besides, the extrema points of a time series could also be identified with our framework.

Although the presented methods are powerful in investigating the deep pattern of a time series, we admit that they are just complements to existing solutions for time series analysis like ARMA-SIN methodology, TBATS methodology, Box-Jenkins methodology, Box-Cox transformation, regression methods, machine learning methods, exponential smoothing method, moving average methods and so on, asserting no dominant position over other methods. This is because each method has its specific application domains, in which no other methods can outperform it.

\section*{Declarations of Interest}
The authors declare that there is no any potential competing interests.

% use section* for acknowledgment
\section*{Acknowledgment}
This work is supported by National Research Foundation of Singapore grant NRF-RSS2016-004 and Ministry of Education Academic Research Fund Tier 1 grants R-266-000-096-133, R-266-000-096-731, R-266-000-100-646 and R-266-000-119-133.
%Shixiong Wang would like to thank Dr. Yue Zhao (Email: yuezhao@u.nus.edu, National University of Singapore) and Dr. Huangjie Zhao (Email: huangjie@u.nus.edu, National University of Singapore) for their helpful discussions in analyzing Eq. (\ref{eq:controbility-check-G1}).

\section*{References}
\bibliographystyle{unsrt}
\bibliography{References.bib}

\begin{thebibliography}{10}

\bibitem{simon2006optimal}
Dan Simon.
\newblock {\em Optimal state estimation: Kalman, H infinity, and nonlinear
  approaches}.
\newblock John Wiley \& Sons, 2006.

\bibitem{hyndman2008forecasting}
Rob Hyndman, Anne~B Koehler, J~Keith Ord, and Ralph~D Snyder.
\newblock {\em Forecasting with exponential smoothing: the state space
  approach}.
\newblock Springer Science \& Business Media, 2008.

\bibitem{hyndman2018forecasting}
Rob~J Hyndman and George Athanasopoulos.
\newblock {\em Forecasting: principles and practice}.
\newblock OTexts, 2018.

\bibitem{box2015time}
George~EP Box, Gwilym~M Jenkins, Gregory~C Reinsel, and Greta~M Ljung.
\newblock {\em Time series analysis: forecasting and control}.
\newblock John Wiley \& Sons, 2015.

\bibitem{engle1982autoregressive}
Robert~F Engle.
\newblock Autoregressive conditional heteroscedasticity with estimates of the
  variance of united kingdom inflation.
\newblock {\em Econometrica: Journal of the Econometric Society}, pages
  987--1007, 1982.

\bibitem{brooks2019introductory}
Chris Brooks.
\newblock {\em Introductory econometrics for finance}.
\newblock Cambridge university press, 2019.

\bibitem{hamilton1995time}
James~D Hamilton.
\newblock Time series analysis.
\newblock {\em Economic Theory. II, Princeton University Press, USA}, pages
  625--630, 1995.

\bibitem{papoulis2002probability}
Athanasios Papoulis and S~Unnikrishna Pillai.
\newblock {\em Probability, random variables, and stochastic processes}.
\newblock Tata McGraw-Hill Education, 2002.

\bibitem{bollerslev1986generalized}
Tim Bollerslev.
\newblock Generalized autoregressive conditional heteroskedasticity.
\newblock {\em Journal of econometrics}, 31(3):307--327, 1986.

\bibitem{posedel2006analysis}
Petra Posedel.
\newblock Analysis of the exchange rate and pricing foreign currency options on
  the croatian market: the ngarch model as an alternative to the black-scholes
  model.
\newblock {\em Financial theory and practice}, 30(4):347--368, 2006.

\bibitem{li2018zd}
Dong Li, Xingfa Zhang, Ke~Zhu, and Shiqing Ling.
\newblock The zd-garch model: A new way to study heteroscedasticity.
\newblock {\em Journal of Econometrics}, 202(1):1--17, 2018.

\bibitem{otto2018generalised}
Philipp Otto, Wolfgang Schmid, and Robert Garthoff.
\newblock Generalised spatial and spatiotemporal autoregressive conditional
  heteroscedasticity.
\newblock {\em Spatial Statistics}, 26:125--145, 2018.

\bibitem{WANG2019ARIMA}
Shixiong Wang, Chongshou Li, and Andrew Lim.
\newblock Why are the arima and sarima not sufficient.
\newblock {\em arXiv preprint arXiv:1904.07632}, 2019.

\bibitem{granger1964spectral}
Clive William~John Granger and Michio Hatanaka.
\newblock {\em Spectral Analysis of Economic Time Series.(PSME-1)}.
\newblock Princeton university press, 1964.

\bibitem{koopmans1995spectral}
Lambert~H Koopmans.
\newblock {\em The spectral analysis of time series}.
\newblock Elsevier, 1995.

\bibitem{bloomfield2004fourier}
Peter Bloomfield.
\newblock {\em Fourier analysis of time series: an introduction}.
\newblock John Wiley \& Sons, 2004.

\bibitem{daubechies1990wavelet}
Ingrid Daubechies.
\newblock The wavelet transform, time-frequency localization and signal
  analysis.
\newblock {\em IEEE transactions on information theory}, 36(5):961--1005, 1990.

\bibitem{huang2014hilbert}
Norden~Eh Huang.
\newblock {\em Hilbert-Huang transform and its applications}, volume~16.
\newblock World Scientific, 2014.

\bibitem{chandra2013square}
Kumar Pakki~Bharani Chandra, Da-Wei Gu, and Ian Postlethwaite.
\newblock Square root cubature information filter.
\newblock {\em IEEE Sensors Journal}, 13(2):750--758, 2013.

\bibitem{li2005survey}
X~Rong Li and Vesselin~P Jilkov.
\newblock Survey of maneuvering target tracking. part v. multiple-model
  methods.
\newblock {\em IEEE Transactions on Aerospace and Electronic Systems},
  41(4):1255--1321, 2005.

\bibitem{durbin2012time}
James Durbin and Siem~Jan Koopman.
\newblock {\em Time series analysis by state space methods}.
\newblock Oxford university press, 2012.

\bibitem{wang2015nonlinear}
Xiaoxu Wang, Yan Liang, Quan Pan, Chunhui Zhao, and Feng Yang.
\newblock Nonlinear gaussian smoothers with colored measurement noise.
\newblock {\em IEEE Transactions on Automatic Control}, 60(3):870--876, 2015.

\bibitem{sarkka2010gaussian}
Simo Sarkka and Jouni Hartikainen.
\newblock On gaussian optimal smoothing of non-linear state space models.
\newblock {\em IEEE Transactions on Automatic Control}, 55(8):1938--1941, 2010.

\bibitem{geng2019state}
Hang Geng, Zidong Wang, Yuhua Cheng, Fuad~E Alsaadi, and Abdullah~M Dobaie.
\newblock State estimation under non-gaussian l{\'e}vy and time-correlated
  additive sensor noises: A modified tobit kalman filtering approach.
\newblock {\em Signal Processing}, 154:120--128, 2019.

\bibitem{liang2004finite}
Yan Liang, Dong~Hua Zhou, Quan Pan, et~al.
\newblock A finite-horizon adaptive kalman filter for linear systems with
  unknown disturbances.
\newblock {\em Signal Processing}, 84(11):2175--2194, 2004.

\bibitem{liu2011kernel}
Weifeng Liu, Jose~C Principe, and Simon Haykin.
\newblock {\em Kernel adaptive filtering: a comprehensive introduction},
  volume~57.
\newblock John Wiley \& Sons, 2011.

\bibitem{ma2017robust}
Wentao Ma, Jiandong Duan, Weishi Man, Haiquan Zhao, and Badong Chen.
\newblock Robust kernel adaptive filters based on mean p-power error for noisy
  chaotic time series prediction.
\newblock {\em Engineering Applications of Artificial Intelligence},
  58:101--110, 2017.

\bibitem{richard2009online}
C{\'e}dric Richard, Jos{\'e} Carlos~M Bermudez, and Paul Honeine.
\newblock Online prediction of time series data with kernels.
\newblock {\em IEEE Transactions on Signal Processing}, 57(3):1058--1067, 2009.

\bibitem{han2018multivariate}
Min Han, Shuhui Zhang, Meiling Xu, Tie Qiu, and Ning Wang.
\newblock Multivariate chaotic time series online prediction based on improved
  kernel recursive least squares algorithm.
\newblock {\em IEEE transactions on cybernetics}, (99):1--13, 2018.

\bibitem{hua2013kernel}
Qu~Hua, Ma~Wen-Tao, Zhao Ji-Hong, and Chen Ba-Dong.
\newblock Kernel least mean kurtosis based online chaotic time series
  prediction.
\newblock {\em Chinese Physics Letters}, 30(11):110505, 2013.

\bibitem{zhu2018power}
Yongli Zhu, Renchang Dai, Guangyi Liu, Zhiwei Wang, and Songtao Lu.
\newblock Power market price forecasting via deep learning.
\newblock In {\em IECON 2018-44th Annual Conference of the IEEE Industrial
  Electronics Society}, pages 4935--4939. IEEE, 2018.

\bibitem{liao2018deep}
Shijie Liao, Jing Chen, Jiaxin Hou, Qingyu Xiong, and Junhao Wen.
\newblock Deep convolutional neural networks with random subspace learning for
  short-term traffic flow prediction with incomplete data.
\newblock In {\em 2018 International Joint Conference on Neural Networks
  (IJCNN)}, pages 1--6. IEEE, 2018.

\bibitem{siami2018forecasting}
Sima Siami-Namini and Akbar~Siami Namin.
\newblock Forecasting economics and financial time series: Arima vs. lstm.
\newblock {\em arXiv preprint arXiv:1803.06386}, 2018.

\bibitem{lu2017agent}
David~W Lu.
\newblock Agent inspired trading using recurrent reinforcement learning and
  lstm neural networks.
\newblock {\em arXiv preprint arXiv:1707.07338}, 2017.

\bibitem{de2011forecasting}
Alysha~M De~Livera, Rob~J Hyndman, and Ralph~D Snyder.
\newblock Forecasting time series with complex seasonal patterns using
  exponential smoothing.
\newblock {\em Journal of the American Statistical Association},
  106(496):1513--1527, 2011.

\bibitem{wang2012stock}
Ju-Jie Wang, Jian-Zhou Wang, Zhe-George Zhang, and Shu-Po Guo.
\newblock Stock index forecasting based on a hybrid model.
\newblock {\em Omega}, 40(6):758--766, 2012.

\bibitem{doya2002multiple}
Kenji Doya, Kazuyuki Samejima, Ken-ichi Katagiri, and Mitsuo Kawato.
\newblock Multiple model-based reinforcement learning.
\newblock {\em Neural computation}, 14(6):1347--1369, 2002.

\bibitem{murray1997multiple}
Roderick Murray-Smith and T~Johansen.
\newblock {\em Multiple model approaches to nonlinear modelling and control}.
\newblock CRC press, 1997.

\bibitem{kreyszig1978introductory}
Erwin Kreyszig.
\newblock {\em Introductory functional analysis with applications}, volume~1.
\newblock wiley New York, 1978.

\bibitem{moghe2019adaptive}
Rahul Moghe, Renato Zanetti, and Maruthi~R Akella.
\newblock Adaptive kalman filter for detectable linear time-invariant systems.
\newblock {\em Journal of Guidance, Control, and Dynamics}, pages 1--9, 2019.

\bibitem{mehra1972approaches}
Raman Mehra.
\newblock Approaches to adaptive filtering.
\newblock {\em IEEE Transactions on automatic control}, 17(5):693--698, 1972.

\bibitem{mohamed1999adaptive}
AH~Mohamed and KP~Schwarz.
\newblock Adaptive kalman filtering for ins/gps.
\newblock {\em Journal of geodesy}, 73(4):193--203, 1999.

\bibitem{myers1976adaptive}
Kenneth Myers and BD~Tapley.
\newblock Adaptive sequential estimation with unknown noise statistics.
\newblock {\em IEEE Transactions on Automatic Control}, 21(4):520--523, 1976.

\bibitem{liang2008adaptive}
Yan Liang, Donghua Zhou, Lei Zhang, and Quan Pan.
\newblock Adaptive filtering for stochastic systems with generalized
  disturbance inputs.
\newblock {\em IEEE Signal Processing Letters}, 15:645--648, 2008.

\bibitem{xia1994adaptive}
Qijun Xia, Ming Rao, Yiqun Ying, and Xuemin Shen.
\newblock Adaptive fading kalman filter with an application.
\newblock {\em Automatica}, 30(8):1333--1338, 1994.

\bibitem{li2003survey}
X~Rong Li and Vesselin~P Jilkov.
\newblock Survey of maneuvering target tracking. part i. dynamic models.
\newblock {\em IEEE Transactions on aerospace and electronic systems},
  39(4):1333--1364, 2003.

\bibitem{roger1994topics}
Horn Roger and R~Johnson Charles.
\newblock {\em Topics in matrix analysis}.
\newblock Cambridge University Press, 1994.

\bibitem{kalman1961new}
Rudolph~E Kalman and Richard~S Bucy.
\newblock New results in linear filtering and prediction theory.
\newblock {\em Journal of basic engineering}, 83(1):95--108, 1961.

\bibitem{anderson1971stability}
BDO Anderson.
\newblock Stability properties of kalman-bucy filters.
\newblock {\em Journal of the Franklin Institute}, 291(2):137--144, 1971.

\bibitem{anderson2012optimal}
Brian~DO Anderson and John~B Moore.
\newblock {\em Optimal filtering}.
\newblock Courier Corporation, 2012.

\bibitem{kim2000robust}
J-H Kim and J-H Oh.
\newblock Robust state estimator of stochastic linear systems with unknown
  disturbances.
\newblock {\em IEE Proceedings-Control Theory and Applications},
  147(2):224--228, 2000.

\bibitem{yong2016unified}
Sze~Zheng Yong, Minghui Zhu, and Emilio Frazzoli.
\newblock A unified filter for simultaneous input and state estimation of
  linear discrete-time stochastic systems.
\newblock {\em Automatica}, 63:321--329, 2016.

\bibitem{diniz2010digital}
Paulo~SR Diniz, Eduardo~AB Da~Silva, and Sergio~L Netto.
\newblock {\em Digital signal processing: system analysis and design}.
\newblock Cambridge University Press, 2010.

\bibitem{li2017approximate}
Tian-cheng Li, Jin-ya Su, Wei Liu, and Juan~M Corchado.
\newblock Approximate gaussian conjugacy: parametric recursive filtering under
  nonlinearity, multimodality, uncertainty, and constraint, and beyond.
\newblock {\em Frontiers of Information Technology \& Electronic Engineering},
  18(12):1913--1939, 2017.

\bibitem{ridolfi2018experimental}
Matteo Ridolfi, Stef Vandermeeren, Jense Defraye, Heidi Steendam, Joeri Gerlo,
  Dirk De~Clercq, Jeroen Hoebeke, and Eli De~Poorter.
\newblock Experimental evaluation of uwb indoor positioning for sport postures.
\newblock {\em Sensors}, 18(1):168, 2018.

\bibitem{karsaz2009optimal}
Ali Karsaz and Hamid Khaloozadeh.
\newblock An optimal two-stage algorithm for highly maneuvering targets
  tracking.
\newblock {\em Signal processing}, 89(4):532--547, 2009.

\end{thebibliography}

% that's all folks
\end{document}